\pdfoutput=1
\documentclass[pdflatex,sn-basic]{sn-jnl}

\usepackage{graphicx}
\usepackage{multirow}
\usepackage{amsmath,amssymb,amsfonts}
\usepackage{amsthm}
\usepackage{mathrsfs}
\usepackage[title]{appendix}
\usepackage{xcolor}
\usepackage{textcomp}
\usepackage{manyfoot}
\usepackage{booktabs}
\usepackage{tikz}
\usetikzlibrary{arrows.meta, positioning, decorations.markings}
\usepackage[section]{placeins}

\theoremstyle{thmstyleone}
\newtheorem{theorem}{Theorem}

\newtheorem{proposition}{Proposition}

\theoremstyle{thmstyletwo}

\theoremstyle{thmstylethree}

\begin{document}

\title[Tractable Estimation of the Money Pump Index]{Tractable Estimation of the Money Pump Index}

\author*[1]{\fnm{Gavin} \sur{Kader}}\email{gav.kader@outlook.com}
\affil[1]{\orgdiv{Center for Intelligence Economics}, \orgname{Southwestern University of Finance and Economics}}

\abstract{The Money Pump Index (MPI) of \citet{echenique2011money} measures the severity of consumer irrationality, but computing the exact mean and median MPI over all revealed preference cycles is NP-hard \citep{smeulders2013note}. Existing solutions rely on heuristic proxies, such as evaluating only shorter cycles or bounding the MPI. By framing revealed preferences as a directed graph of expenditures, this paper projects choice violations onto weakly fundamental directed cycle bases, which are minimal sets of linearly independent cycles that span the graph's entire cycle space, offering an equivalent representation of the mean and median MPI. Randomly sampling over the possible cycle bases gives rise to computationally tractable, consistent estimators for the exact mean and median MPI. Applied to the scanner dataset analyzed by \citet{echenique2011money} and \citet{smeulders2013note}, the proposed estimators compute quickly and with negligible small-sample bias.}


\maketitle

\section{Introduction}

Irrational behavior makes consumers vulnerable because it allows arbitrageurs to extract money from inconsistent choices. \citet{echenique2011money} operationalized this concept by proposing the Money Pump Index (MPI), which measures the fraction of a consumer's budget an arbitrageur could make by exploiting violations of the generalized axiom of revealed preference (GARP). When a consumer's choices violate GARP, there are typically multiple choice cycles that generate money pumps. A natural way to summarize the severity of a consumer's overall irrationality is to take the mean or median value of the money pumps across all such cycles.

However, calculating this mean or median severity is computationally intractable for large datasets. For example, consider a dataset of prices and bundles $D = \{(\mathbf{p}^t, \mathbf{x}^t)\}_{t=1}^n$ with $n=25$ observations. The number of potential cycles of length $k$ is $\binom{n}{k}(k-1)!$. For instance, if we just look at cycles of length $k=10$, there are $\binom{25}{10} \times 9! \approx 1.18 \times 10^{12}$ potential cycles. Summing across all possible cycle lengths results in an inordinately large number of cycles. Indeed, \citet{smeulders2013note} formalize this computational challenge by proving that computing the exact mean MPI (OA-MPI) is an NP-hard problem. Consequently, researchers analyzing datasets with many observations often use approximations to the true mean and median MPI by considering only shorter cycles, or using instead extreme values like the maximum/minimum MPI \citep{smeulders2013note}.

I examine whether the mean and median MPI can be tractably measured without imposing arbitrary cycle length restrictions. I overcome the $n!$ scaling associated with evaluating all cycles (i.e., the $\mathcal{O}(n!)$ complexity) by leveraging graph theory. Specifically, a dataset $D$ can be reformulated as a complete directed graph of expenditures, where observations are nodes and every directed edge represents a pairwise expenditure evaluation (e.g., $\mathbf{p}^i \cdot \mathbf{x}^j$). While actual revealed preference relations are inherently incomplete (existing only when a bundle is affordable), this underlying expenditure graph is naturally complete. True revealed preference cycles---specifically, GARP violations---exist simply as subgraphs within this complete structure where the affordability condition is satisfied along every edge. By adopting this framework, I propose a new approach that projects choice violations onto the graph's \textit{weakly fundamental directed cycle basis}---a minimal, linearly independent set of cycles that spans the entire cycle space of the graph.\footnote{Recent work by \citet{caradonna2026inconsistency} similarly recognizes the dependent structure of choice cycles, proposing that the cardinality of a basis of independent cycles (the ``inconsistency rank'') serves as a natural measure of irrationality that avoids double-counting forced cycles. While the focus of \citet{caradonna2026inconsistency} is on the semantic interpretation of the basis itself as a measure of inconsistency, this paper utilizes weakly fundamental directed cycle bases as a computational mechanism to tractably estimate the aggregate MPI over the entire cycle space.}

I introduce the Basis Average MPI, an equivalent formulation of the true mean MPI that aggregates over the graph's weakly fundamental directed cycle bases rather than enumerating all individual cycles. While the total number of possible cycle bases remains large, this reformulation gives rise to an efficient estimator. By randomly sampling over the space of weakly fundamental directed cycle bases, the estimator approximates the true mean MPI in polynomial time ($\mathcal{O}(K n^2)$, where $K$ is the sample size and $n$ is the number of time periods/observations). I formally establish the consistency of this estimator and provide theoretical bounds on its finite-sample bias, which proves negligible in practice. Analogously, I introduce the Basis Percentile MPI and a corresponding consistent estimator for the median MPI (and any arbitrary percentile). Finally, I apply this methodology to the Stanford Basket Dataset as in \citet{echenique2011money} and \citet{smeulders2013note}. Across all GARP-violating households (396), the estimators compute in seconds, achieving near-zero empirical bias. Furthermore, both the estimated basis mean and basis median MPI closely align with the short-cycle approximation of \citet{echenique2011money}, formally confirming that truncating at shorter cycle lengths ($k \le 4$) introduces only negligible bias while providing an estimator that scales tractably to choice environments with more observations.

The remainder of the paper proceeds as follows. Section \ref{sec:garp_mpi} formally links the MPI to the directed graph representation. Section \ref{sec:ensemble} constructs the weakly fundamental directed cycle bases and proves the consistency of the random-path estimators. Section \ref{sec:empirical} validates the estimators on the Stanford Basket Dataset, demonstrating near-zero empirical bias. All proofs are relegated to the appendix.

\section{GARP and the MPI} \label{sec:garp_mpi}

\textbf{GARP}: Suppose we have a dataset $D = \{(\mathbf{p}^t, \mathbf{x}^t)\}_{t=1}^n$ of $n$ observations, where $\mathbf{p}^t \in \mathbb{R}_{++}^L$ is a positive price vector and $\mathbf{x}^t \in \mathbb{R}_+^L \setminus \{\mathbf{0}\}$ is a non-negative, non-zero consumption bundle (ensuring strictly positive expenditures $E_t = \mathbf{p}^t \cdot \mathbf{x}^t > 0$ for all $t$). For any two observations $i$ and $j$, if and only if $\mathbf{p}^i \cdot \mathbf{x}^i \ge \mathbf{p}^i \cdot \mathbf{x}^j$, we say that $\mathbf{x}^i$ is \textit{directly revealed preferred} to $\mathbf{x}^j$, denoted as $\mathbf{x}^i R_0 \mathbf{x}^j$. The indirect revealed preference relation, denoted by $R$, is the transitive closure of $R_0$. Furthermore, $\mathbf{x}^i$ is \textit{strictly directly revealed preferred} to $\mathbf{x}^j$ if and only if $\mathbf{p}^i \cdot \mathbf{x}^i > \mathbf{p}^i \cdot \mathbf{x}^j$, denoted as $\mathbf{x}^i P_0 \mathbf{x}^j$. A dataset $D$ satisfies the generalized axiom of revealed preference (GARP) \citep{varian1982nonparametric, varian1983non} if, for any sequence of observations $(k_1, k_2, \dots, k_m)$ such that $\mathbf{x}^{k_1} R_0 \mathbf{x}^{k_2}$, $\mathbf{x}^{k_2} R_0 \mathbf{x}^{k_3}$, $\dots$, $\mathbf{x}^{k_{m-1}} R_0 \mathbf{x}^{k_m}$, it is not the case that $\mathbf{x}^{k_m} P_0 \mathbf{x}^{k_1}$.

As GARP is a binary test, we frequently use indices to quantify the severity of irrationality. Alongside the MPI \citep{echenique2011money}, the literature has proposed numerous such measures, such as Afriat's Efficiency Index \citep{afriat1967construction}, Varian's Index \citep{varian1990goodness}, the Houtman-Maks Index \citep{houtman1985determining}, the Minimum Cost Index \citep{dean2016measuring}, the Swaps Index \citep{apesteguia2015measure}, etc.\footnote{\citet{demuynck2023computing} introduced Mixed-Integer Linear Programming (MILP) formulations for the average Varian Index, the Houtman-Maks Index, and the Minimum Cost Index. MILP solvers are typically very fast, overcoming computational problems normally associated with those indices.}

\textbf{MPI}: The fundamental idea behind the MPI is one of ``exploiting the vulnerability'' of irrational consumers by measuring the amount of money an ``arbitrageur'' could extract from them. If a consumer (represented by their dataset) violates GARP, it is because they exhibit a revealed preference cycle from their choices. For example, suppose observations $1$ and $2$ form a 2-cycle (cycle of length 2) such that $\mathbf{p}^1 \cdot \mathbf{x}^1 \ge \mathbf{p}^1 \cdot \mathbf{x}^2$ and $\mathbf{p}^2 \cdot \mathbf{x}^2 > \mathbf{p}^2 \cdot \mathbf{x}^1$, i.e., they chose such that $\mathbf{x}^1 R_0 \mathbf{x}^2$, and $\mathbf{x}^2 P_0 \mathbf{x}^1$.

As the consumer chose $\mathbf{x}^1$ at prices $\mathbf{p}^1$ when $\mathbf{x}^2$ was affordable, they implicitly value ``upgrading'' from $\mathbf{x}^2$ to $\mathbf{x}^1$ at a minimum of $\mathbf{p}^1 \cdot (\mathbf{x}^1 - \mathbf{x}^2)$. Likewise, at prices $\mathbf{p}^2$, the minimum value of upgrading from $\mathbf{x}^1$ to $\mathbf{x}^2$ is $\mathbf{p}^2 \cdot (\mathbf{x}^2 - \mathbf{x}^1)$. As described, an arbitrageur makes a profit by providing this ``exploitative upgrade service'' to the consumer, with arbitrageur profit being $\mathbf{p}^1 \cdot (\mathbf{x}^1 - \mathbf{x}^2) + \mathbf{p}^2 \cdot (\mathbf{x}^2 - \mathbf{x}^1)$, while still leaving the consumer with their preferred bundles. If it were not the case that $\mathbf{x}^2 P_0 \mathbf{x}^1$, the arbitrageur is no longer guaranteed to make a profit (i.e., there is no irrationality to exploit). 

Generalizing to an arbitrary preference cycle $C = (k_1, k_2, \ldots, k_m)$ of length $m$, we define $\text{MPI}_D(C)$ as the arbitrageur's total money pump profit relative to the consumer's total expenditure within that cycle (from dataset $D$):
\begin{equation}
    \text{MPI}_D(C) = \frac{\sum_{i=1}^m \mathbf{p}^{k_i} \cdot (\mathbf{x}^{k_i} - \mathbf{x}^{k_{i+1}})}{\sum_{i=1}^m \mathbf{p}^{k_i} \cdot \mathbf{x}^{k_i}} = 1 - \frac{\sum_{i=1}^m \mathbf{p}^{k_i} \cdot \mathbf{x}^{k_{i+1}}}{\sum_{i=1}^m \mathbf{p}^{k_i} \cdot \mathbf{x}^{k_i}},
\end{equation}
where $k_{m+1} = k_1$. For a given cycle $C$, this captures the proportion of the total cycle expenditure an arbitrageur could extract. To measure the MPI for an entire dataset $D$, \citet{echenique2011money} recommend the mean (denoted $\text{OA-MPI}_D$) or median (denoted $\text{OP}_{50}\text{-MPI}_D$, or more generally any $\alpha$-percentile $\text{OP}_\alpha\text{-MPI}_D$) of $\text{MPI}_D(C)$ over all cycles within $D$, while \citet{smeulders2013note} adapt the minimum cycle mean algorithm of \citet{karp1978characterization} to compute the maximum and minimum of $\text{MPI}_D(C)$ over all cycles (providing informative bounds on the OA-MPI).

We represent the dataset $D$ as a complete directed expenditure graph $G_D = (V_D, E_D)$, where each node represents an observation $t \in \{1, \dots, n\}$. For any two observations $u$ and $v$, the directed edge $(u, v)$ has an associated money pump surplus $s_{u,v} = \mathbf{p}^u \cdot (\mathbf{x}^u - \mathbf{x}^v)$ and period expenditure $E_u = \mathbf{p}^u \cdot \mathbf{x}^u$, so that the numerator and denominator of $\text{MPI}_D(C)$ sum $s_{k_i, k_{i+1}}$ and $E_{k_i}$ along the cycle's edges. An edge satisfies affordability if $s_{u,v} \geq 0$ (meaning bundle $\mathbf{x}^v$ was affordable at prices $\mathbf{p}^u$, so $\mathbf{x}^u R_0 \mathbf{x}^v$), and represents a strict direct revealed preference if $s_{u,v} > 0$ ($\mathbf{x}^u P_0 \mathbf{x}^v$), making it straightforward to check from this graph whether any cycle violates GARP.\footnote{Crucially, this pre-computed matrix makes it computationally trivial to evaluate whether any generic cycle drawn from the expenditure graph constitutes a true GARP violation: we simply verify that $s_{u,v} \geq 0$ for every edge in the cycle and that $s_{u,v} > 0$ for at least one edge. If even one edge has $s_{u,v} < 0$, the cycle does not violate GARP.}

\bigskip

\noindent\textbf{Example:} To solidify understanding, consider a consumer with dataset $D$ comprising three observations over three distinct goods, with strictly positive consumption bundles $\mathbf{x}^1 = (1, 1, 2)$, $\mathbf{x}^2 = (1, 2, 1)$, and $\mathbf{x}^3 = (1, 1, 3)$, at prices $\mathbf{p}^1 = (2, 1, 4)$, $\mathbf{p}^2 = (2, 3, 1)$, and $\mathbf{p}^3 = (2, 4, 3)$ with expenditures $E_1 = \mathbf{p}^1 \cdot \mathbf{x}^1 = 11$, $E_2 = \mathbf{p}^2 \cdot \mathbf{x}^2 = 9$, and $E_3 = \mathbf{p}^3 \cdot \mathbf{x}^3 = 15$. 

Computing pairwise surpluses $s_{u,v} = \mathbf{p}^u \cdot (\mathbf{x}^u - \mathbf{x}^v)$ yields:
\begin{align*}
    s_{1,2} &= 11 - (2 + 2 + 4) = 3 > 0, & s_{2,1} &= 9 - (2 + 3 + 2) = 2 > 0, \\[0.85em]
    s_{2,3} &= 9 - (2 + 3 + 3) = 1 > 0, & s_{3,2} &= 15 - (2 + 8 + 3) = 2 > 0, \\[0.85em]
    s_{1,3} &= 11 - (2 + 1 + 12) = -4 < 0, & s_{3,1} &= 15 - (2 + 4 + 6) = 3 > 0.
\end{align*}
The dataset exhibits three valid GARP-violating preference cycles: i) $\mathbf{x}^1 P_0 \mathbf{x}^2$ and $\mathbf{x}^2 P_0 \mathbf{x}^1$; ii) $\mathbf{x}^2 P_0 \mathbf{x}^3$ and $\mathbf{x}^3 P_0 \mathbf{x}^2$; and iii) $\mathbf{x}^1 P_0 \mathbf{x}^2$, $\mathbf{x}^2 P_0 \mathbf{x}^3$, and $\mathbf{x}^3 P_0 \mathbf{x}^1$.

Let $C_{k_1 k_2 \dots k_j}$ denote the choice cycle associated with the sequence of observations $k_1, k_2, \dots, k_j, k_1$ (with a $j$-cycle referring to a generic cycle of length $j$). An arbitrageur can only extract a positive money pump from cycles that strictly violate GARP. For any generic cycle $C$, let $V$ denote the set of valid GARP-violating cycles (where every directed edge satisfies affordability and at least one edge represents a strict direct preference). We define $\text{MPI}_D(C)$ such that non-violating cycles yield zero extracted surplus:
\begin{equation}
\text{MPI}_D(C) = \begin{cases}
\dfrac{\sum_{i=1}^{|C|} \mathbf{p}^{k_i} \cdot (\mathbf{x}^{k_i} - \mathbf{x}^{k_{i+1}})}{\sum_{i=1}^{|C|} \mathbf{p}^{k_i} \cdot \mathbf{x}^{k_i}} & \text{if } C \in V, \\[1em]
0 & \text{if } C \notin V.
\end{cases}
\end{equation}
For our three valid cycles, the MPIs evaluate to:
\begin{align*}
    \text{MPI}_{D}(C_{12}) &= \frac{3 + 2}{11 + 9} = \frac{5}{20} = 0.250,  \\[0.85em]
    \text{MPI}_{D}(C_{23}) &= \frac{1 + 2}{9 + 15} = \frac{3}{24} = 0.125, \\[0.85em]
    \text{MPI}_{D}(C_{123}) &= \frac{3 + 1 + 3}{11 + 9 + 15} = \frac{7}{35} = 0.200. 
\end{align*}
The standard OA-MPI, as recommended by \citet{echenique2011money}, measures the global severity of irrationality by averaging the money pump costs exclusively over the valid GARP-violating cycles in the dataset yielding 
\begin{align*}
    \text{OA-MPI}_{D} &= \frac{1}{3} (0.250 + 0.125 + 0.200) = \frac{0.575}{3} \approx 0.1917.
\end{align*}

For a dataset with $n$ observations, the total number of distinct potential choice cycles is given by $\sum_{j=2}^{n} \binom{n}{j} (j-1)!$, which clearly explodes as $n$ increases (with $\mathcal{O}(n!)$ complexity). Consequently, even when restricting attention to the revealed preference subgraph, enumerating all GARP-violating cycles in $V$ to evaluate this average becomes practically impossible for larger datasets. Indeed, as shown by \citet{smeulders2013note}, computing the true mean (or median) MPI is an NP-hard problem, meaning that there can only exist a polynomial time algorithm for computation if $\text{P}=\text{NP}$.

\section{Basis MPI} \label{sec:ensemble}

\noindent\textbf{Example (continued):} I show that we can select a specific subset of just four cycles---referred to as a \textit{weakly fundamental directed cycle basis}---to construct all other cycles. To see this, we use a directed \textit{Hamiltonian path} (hereafter, \textit{H-path}), which is a directed path that visits every node in a graph precisely once. By ordering the vertices along this path, we can classify all other edges in the graph relative to this sequence to form a cycle basis.

Consider the path sequence $l_1: 1 \to 2 \to 3$. We extract its fundamental cycle basis, denoted $\mathcal{B}_{l_1}$, as follows:
\begin{enumerate}
    \item \textbf{Contiguous subsegments:} A subsegment is a contiguous sequence of edges directly along the H-path. We take all possible contiguous subsegments of the path and close each one using a single \textit{back-edge} (an edge that points backwards against the path sequence) to form a cycle. For example, the subsegment $(1 \to 2)$ is closed by the back-edge $2\to1$ yielding $C_{12}$; $(2 \to 3)$ is closed by $3\to2$ yielding $C_{23}$; and the full path $(1 \to 2 \to 3)$ is closed by $3\to1$ yielding $C_{123}$.
    \item \textbf{Skip-Edge 2-Cycles:} A \textit{skip-edge} is an edge that skips ahead along the H-path, jumping over intermediate vertices. We take any skip-edges along the path and close them using a single back-edge to form a 2-cycle. In our path, the only forward skip-edge is $(1 \to 3)$, which we close with the back-edge $3\to1$ yielding $C_{13}$.
\end{enumerate}
This yields the weakly fundamental directed cycle basis $\mathcal{B}_{l_1} = \{C_{12}, C_{23}, C_{123}, C_{13}\}$ as shown in Figure \ref{fig:permutation_basis}.\footnote{Strictly speaking, in graph-theoretic terminology \citep{kavitha2009cycle}, because skip-edge 2-cycles contain two chords relative to the Hamiltonian path, this basis is a \textit{weakly fundamental directed cycle basis}. For exposition and brevity throughout the remainder of this paper, I refer to it as a \textit{fundamental cycle basis} (or \textit{cycle basis}) interchangeably.}

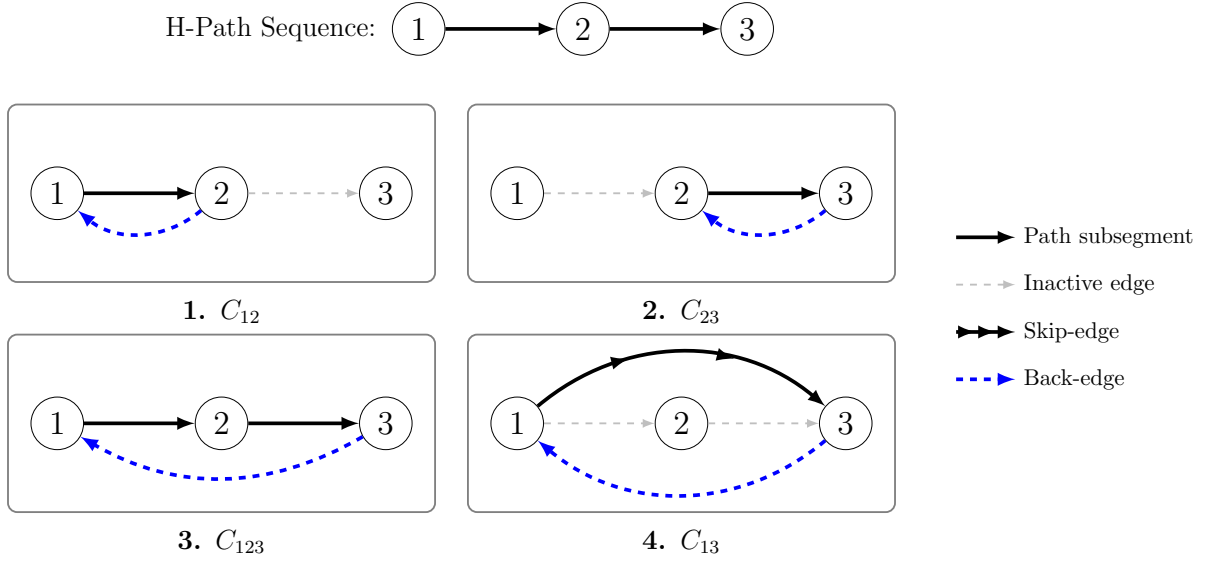
\begin{figure}[ht]
    \centering
    \resizebox{\textwidth}{!}{
\begin{tikzpicture}[
    vertex/.style={circle, draw, minimum size=8mm, font=\large},
    treeedge/.style={->, ultra thick, >=latex},
    inactiveedge/.style={->, thick, >=latex, dashed, gray!50},
    cycleedge/.style={->, ultra thick, >=latex, dashed, color=blue},
    skipedge/.style={->, ultra thick, >=latex, postaction={decorate, decoration={markings, mark=at position 0.33 with {\arrow{latex}}, mark=at position 0.67 with {\arrow{latex}}}}},
    box/.style={gray, thick, rounded corners}
]

% Hamiltonian path at top
\begin{scope}[xshift=5.5cm, yshift=2.5cm]
    \node[anchor=east] at (-0.5, 0) {H-Path Sequence:};
    \node[vertex] (v1) at (0, 0) {1};
    \node[vertex] (v2) at (2.5, 0) {2};
    \node[vertex] (v3) at (5, 0) {3};
    \draw[treeedge] (v1) -- (v2);
    \draw[treeedge] (v2) -- (v3);
\end{scope}

% 1. C_12
\begin{scope}[xshift=0cm, yshift=0cm]
    \draw[box] (-0.75, 1.35) rectangle (5.75, -1.35);
    \node[vertex] (v1) at (0, 0) {1};
    \node[vertex] (v2) at (2.5, 0) {2};
    \node[vertex] (v3) at (5, 0) {3};
    \draw[treeedge] (v1) -- (v2);
    \draw[inactiveedge] (v2) -- (v3);
    \draw[cycleedge] (v2) to[bend left=40] (v1);
    \node at (2.5, -1.8) {\textbf{1. $C_{12}$}};
\end{scope}

% 2. C_23
\begin{scope}[xshift=7cm, yshift=0cm]
    \draw[box] (-0.75, 1.35) rectangle (5.75, -1.35);
    \node[vertex] (v1) at (0, 0) {1};
    \node[vertex] (v2) at (2.5, 0) {2};
    \node[vertex] (v3) at (5, 0) {3};
    \draw[inactiveedge] (v1) -- (v2);
    \draw[treeedge] (v2) -- (v3);
    \draw[cycleedge] (v3) to[bend left=40] (v2);
    \node at (2.5, -1.8) {\textbf{2. $C_{23}$}};
\end{scope}

% 3. C_123
\begin{scope}[xshift=0cm, yshift=-3.5cm]
    \draw[box] (-0.75, 1.35) rectangle (5.75, -1.35);
    \node[vertex] (v1) at (0, 0) {1};
    \node[vertex] (v2) at (2.5, 0) {2};
    \node[vertex] (v3) at (5, 0) {3};
    \draw[treeedge] (v1) -- (v2);
    \draw[treeedge] (v2) -- (v3);
    \draw[cycleedge] (v3) to[bend left=30] (v1);
    \node at (2.5, -1.8) {\textbf{3. $C_{123}$}};
\end{scope}

% 4. C_13
\begin{scope}[xshift=7cm, yshift=-3.5cm]
    \draw[box] (-0.75, 1.35) rectangle (5.75, -1.35);
    \node[vertex] (v1) at (0, 0) {1};
    \node[vertex] (v2) at (2.5, 0) {2};
    \node[vertex] (v3) at (5, 0) {3};
    \draw[inactiveedge] (v1) -- (v2);
    \draw[inactiveedge] (v2) -- (v3);
    \draw[skipedge] (v1) to[bend left=40] (v3);
    \draw[cycleedge] (v3) to[bend left=40] (v1);
    \node at (2.5, -1.8) {\textbf{4. $C_{13}$}};
\end{scope}

% Legend
\begin{scope}[xshift=13.5cm, yshift=-1.75cm, scale=0.9, transform shape]
    
    \draw[treeedge] (0.2, 1.2) -- (1.2, 1.2);
    \node[anchor=west, font=\small] at (1.2, 1.2) {Path subsegment};
    
    \draw[inactiveedge] (0.2, 0.4) -- (1.2, 0.4);
    \node[anchor=west, font=\small, text=black] at (1.2, 0.4) {Inactive edge};
    
    \draw[skipedge] (0.2, -0.4) -- (1.2, -0.4);
    \node[anchor=west, font=\small] at (1.2, -0.4) {Skip-edge};
    
    \draw[cycleedge] (0.2, -1.2) -- (1.2, -1.2);
    \node[anchor=west, font=\small, text=black] at (1.2, -1.2) {Back-edge};
\end{scope}

\end{tikzpicture}
}
    \caption{Extracting the cycle basis from the H-path $1 \to 2 \to 3$. Each fundamental cycle is formed by closing a forward path (or skip-edge) with exactly one back-edge.}
    \label{fig:permutation_basis}
\end{figure}

As we frame our setting as a complete directed graph of expenditures, there are exactly $n!$ distinct H-paths, generating an ensemble of $n!$ path-indexed cycle bases. For our 3-observation example, there are $3! = 6$ paths, yielding 6 path-indexed cycle bases as follows:
\begin{align*}
    l_1 (1\to 2\to 3): \mathcal{B}_{l_1} &= \{C_{12}, C_{23}, C_{123}, C_{13}\}, \\[0.85em]
    l_2 (1\to 3\to 2): \mathcal{B}_{l_2} &= \{C_{13}, C_{32}, C_{132}, C_{12}\}, \\[0.85em]
    l_3 (2\to 1\to 3): \mathcal{B}_{l_3} &= \{C_{21}, C_{13}, C_{132}, C_{23}\}, \\[0.85em]
    l_4 (2\to 3\to 1): \mathcal{B}_{l_4} &= \{C_{23}, C_{31}, C_{123}, C_{21}\}, \\[0.85em]
    l_5 (3\to 1\to 2): \mathcal{B}_{l_5} &= \{C_{31}, C_{12}, C_{123}, C_{32}\}, \\[0.85em]
    l_6 (3\to 2\to 1): \mathcal{B}_{l_6} &= \{C_{32}, C_{21}, C_{132}, C_{31}\}.  
\end{align*}

As a cycle basis spans the entire cycle space, any cycle not in the basis can be constructed through a linear combination of basis cycles. For example, Figure \ref{fig:linear_combination} illustrates how $C_{132}$ can be constructed using cycles strictly from $\mathcal{B}_{l_1}$ via the cycle formed from the linear combination $C_{132} = C_{13} + C_{12} + C_{23} - C_{123}$.\footnote{For 2-cycles, their ordering is unique up to cyclic permutation (i.e., interchangeable) meaning $C_{12} \equiv C_{21}$, $C_{13} \equiv C_{31}$, $C_{23} \equiv C_{32}$, and thus their associated MPIs are identical.}

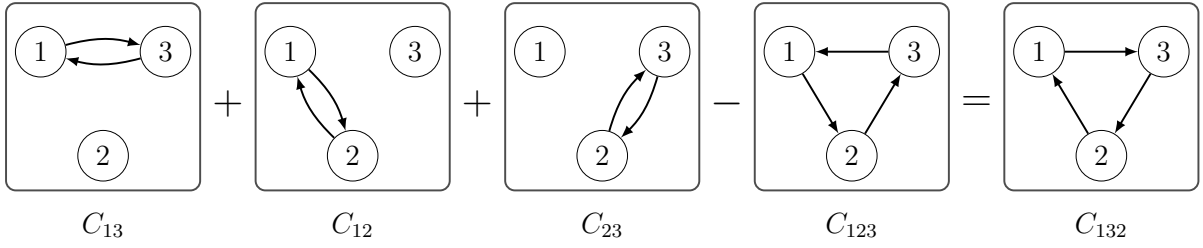
\begin{figure}[ht]
    \centering
    \resizebox{\textwidth}{!}{
\begin{tikzpicture}[
    vertex/.style={circle, draw, minimum size=6mm, font=\normalsize},
    edge/.style={->, thick, >=latex},
    box/.style={draw, rounded corners, thick, draw=black!70}
]

\begin{scope}[xshift=0cm, yshift=0cm]
    \draw[box] (-0.5, 0.7) rectangle (2.3, -2.0);
    \node[vertex] (v1) at (0, 0) {1};
    \node[vertex] (v3) at (1.8, 0) {3};
    \node[vertex] (v2) at (0.9, -1.5) {2};
    \draw[edge] (v1) to[bend left=15] (v3);
    \draw[edge] (v3) to[bend left=15] (v1);
    \node at (0.9, -2.5) {$C_{13}$};
\end{scope}

\node at (2.7, -0.75) {\Large $+$};

\begin{scope}[xshift=3.6cm, yshift=0cm]
    \draw[box] (-0.5, 0.7) rectangle (2.3, -2.0);
    \node[vertex] (v1) at (0, 0) {1};
    \node[vertex] (v3) at (1.8, 0) {3};
    \node[vertex] (v2) at (0.9, -1.5) {2};
    \draw[edge] (v1) to[bend left=15] (v2);
    \draw[edge] (v2) to[bend left=15] (v1);
    \node at (0.9, -2.5) {$C_{12}$};
\end{scope}

\node at (6.3, -0.75) {\Large $+$};

\begin{scope}[xshift=7.2cm, yshift=0cm]
    \draw[box] (-0.5, 0.7) rectangle (2.3, -2.0);
    \node[vertex] (v1) at (0, 0) {1};
    \node[vertex] (v3) at (1.8, 0) {3};
    \node[vertex] (v2) at (0.9, -1.5) {2};
    \draw[edge] (v2) to[bend left=15] (v3);
    \draw[edge] (v3) to[bend left=15] (v2);
    \node at (0.9, -2.5) {$C_{23}$};
\end{scope}

\node at (9.9, -0.75) {\Large $-$};

\begin{scope}[xshift=10.8cm, yshift=0cm]
    \draw[box] (-0.5, 0.7) rectangle (2.3, -2.0);
    \node[vertex] (v1) at (0, 0) {1};
    \node[vertex] (v3) at (1.8, 0) {3};
    \node[vertex] (v2) at (0.9, -1.5) {2};
    \draw[edge] (v1) -- (v2);
    \draw[edge] (v2) -- (v3);
    \draw[edge] (v3) -- (v1);
    \node at (0.9, -2.5) {$C_{123}$};
\end{scope}

\node at (13.5, -0.75) {\Large $=$};

\begin{scope}[xshift=14.4cm, yshift=0cm]
    \draw[box] (-0.5, 0.7) rectangle (2.3, -2.0);
    \node[vertex] (v1) at (0, 0) {1};
    \node[vertex] (v3) at (1.8, 0) {3};
    \node[vertex] (v2) at (0.9, -1.5) {2};
    \draw[edge] (v1) -- (v3);
    \draw[edge] (v3) -- (v2);
    \draw[edge] (v2) -- (v1);
    \node at (0.9, -2.5) {$C_{132}$};
\end{scope}

\end{tikzpicture}
}
    \caption{Construction of the non-basis cycle $C_{132}$ using a linear combination of cycles from $\mathcal{B}_{l_1}$.}
    \label{fig:linear_combination}
\end{figure}

To evaluate the total money pump cost of a given cycle basis, we aggregate the MPIs of its constituent cycles. As the basis is drawn from the complete expenditure graph, it contains both actual revealed preference cycles and non-violating cycles. By filtering these through the true revealed preference relations and assigning an MPI of 0 to any cycle that does not strictly violate GARP, we ensure that we only penalize genuine choice inconsistencies. 

For each cycle basis, it might seem intuitive to take the arithmetic mean of the MPIs of the cycles in that cycle basis, and then aggregating again by taking the mean across all six cycle bases to arrive at a global measure for the MPI over the dataset. However, this is not quite the right approach as cycles of different lengths are essentially oversampled at different rates. By simply counting the cycles across $l_1$ to $l_6$, we can see that there are eighteen 2-cycles and six 3-cycles. Specifically, there are three unique 2-cycles ($C_{12}, C_{13}, C_{23}$) that appear 6 times each, whereas the two unique 3-cycles ($C_{123}, C_{132}$) appear only 3 times each. We need to account for these frequencies when aggregating. 

To achieve this, we must instead compute a weighted average where the MPI of a given cycle is weighted inversely by its frequency. For our 3-observation dataset, we weight the MPIs of 2-cycles by $1/6$ and 3-cycles by $1/3$. Recall that we previously identified exactly three valid GARP-violating cycles in this dataset ($C_{12}, C_{23}, C_{123}$) with MPI values of $0.250$, $0.125$, and $0.200$ respectively. For any given basis $l_i$, we compute its frequency-adjusted sum of violating cycle costs $M(l_i)$, alongside a frequency-adjusted weight $W(l_i)$ (where non-violating cycles similarly contribute a weight of 0):
\begin{align*}
    l_1: \quad M(l_1) &= \frac{0.250}{6} + \frac{0.125}{6} + \frac{0.200}{3}, \quad & W(l_1) &= \frac{1}{6} + \frac{1}{6} + \frac{1}{3}, \\[0.5em]
    l_2: \quad M(l_2) &= \frac{0.125}{6} + \frac{0.250}{6}, \quad & W(l_2) &= \frac{1}{6} + \frac{1}{6}, \\[0.5em]
    l_3: \quad M(l_3) &= \frac{0.250}{6} + \frac{0.125}{6}, \quad & W(l_3) &= \frac{1}{6} + \frac{1}{6}, \\[0.5em]
    l_4: \quad M(l_4) &= \frac{0.125}{6} + \frac{0.200}{3} + \frac{0.250}{6}, \quad & W(l_4) &= \frac{1}{6} + \frac{1}{3} + \frac{1}{6}, \\[0.5em]
    l_5: \quad M(l_5) &= \frac{0.250}{6} + \frac{0.200}{3} + \frac{0.125}{6}, \quad & W(l_5) &= \frac{1}{6} + \frac{1}{3} + \frac{1}{6}, \\[0.5em]
    l_6: \quad M(l_6) &= \frac{0.125}{6} + \frac{0.250}{6}, \quad & W(l_6) &= \frac{1}{6} + \frac{1}{6}. 
\end{align*}

To calculate what I refer to as the \textit{Basis Average MPI} for the dataset $D$ ($\text{BA-MPI}_D$), we sum these adjusted costs across the entire ensemble of bases and normalize by the sum of their weights:
\begin{equation*}
    \text{BA-MPI}_D = \frac{\sum_{i=1}^{6} M(l_i)}{\sum_{i=1}^{6} W(l_i)} = \frac{0.575}{3} \approx 0.1917.
\end{equation*}
This evaluates \textit{exactly} to the previously calculated $\text{OA-MPI}_{D}$. Filtering out cycles that do not violate GARP and weighting GARP-violating cycles from the weakly fundamental directed cycle bases by their appropriate frequencies across all bases allows us to recover the original mean MPI suggested by \citet{echenique2011money}.

\subsection{Basis Average MPI}
The logic from the previous example extends to arbitrary finite datasets, from which we can provide a highly tractable method for calculating the exact OA-MPI for any dataset. 

First, to obtain the weakly fundamental directed cycle basis from a generic H-path, we divide the available edges in the complete directed expenditure graph into the same two categories of \textit{forward-edges}, which jump ahead along the path sequence, and \textit{back-edges}, which jump backwards. Every cycle in the basis is formed by closing a forward path with exactly one back-edge. Cycles from contiguous subsegments of the H-path, say $v_i \to v_{i+1} \to \dots \to v_j$, are simply formed by taking the subsegment and closing the loop using a single back-edge from $v_j \to v_i$. This set of cycles extracts $n(n-1)/2$ basis cycles. Similarly, for the skip-edge 2-cycles, we take every forward skip-edge (e.g., $v_i \to v_j$ where $j \ge i+2$) and close the loop using the single back-edge $v_j \to v_i$. This extracts an additional $(n-1)(n-2)/2$ basis cycles. Thus, for each H-path we have a weakly fundamental directed cycle basis with exactly $(n-1)^2$ distinct cycles.\footnote{Linear independence of the $(n-1)^2$ basis cycles follows directly from their edge incidence structure. The complete directed graph has $(n-1)^2$ non-path edges: $\binom{n-1}{2}$ forward skip-edges ($v_i \to v_j$ with $j \ge i+2$) and $\binom{n}{2}$ back-edges ($v_j \to v_i$ with $j > i$). We partition the extracted cycles into the $\binom{n-1}{2}$ skip-edge 2-cycles and the $\binom{n}{2}$ contiguous subsegment cycles (which include the adjacent 2-cycles). Notice that each skip-edge 2-cycle contains uniquely one forward skip-edge and one closing back-edge, while each contiguous subsegment cycle contains only path edges and uniquely one back-edge. Ordering forward skip-edges before back-edges for the columns, and skip-edge 2-cycles before contiguous cycles for the rows, the cycle-edge incidence matrix takes the block upper-triangular form $\bigl(\begin{smallmatrix} \hat{I} & A \\ 0 & \tilde{I} \end{smallmatrix}\bigr)$, where $\hat{I}$ and $\tilde{I}$ are identity matrices of order $\binom{n-1}{2}$ and $\binom{n}{2}$, respectively, and $A$ is the back-edge incidence submatrix for skip-edge 2-cycles. Since the diagonal blocks are identity matrices, the matrix has full row rank $(n-1)^2$, guaranteeing strict linear independence.}

Given the fundamental cycle basis construction, we saw in the previous example the need to weight by the inverse of the exact frequencies of $j$-cycles, as shorter cycles are essentially oversampled relative to longer ones across the ensemble of bases. To formally derive these frequencies, let $\mathbb{I}(C \in \mathcal{B}_p)$ be an indicator function equal to $1$ if cycle $C$ is generated by path $p$. The total number of times $C$ is extracted across all $n!$ possible bases is $\sum_{p \in \mathcal{P}} \mathbb{I}(C \in \mathcal{B}_p)$. Given the way we construct the fundamental cycle bases, we can derive an exact formula for $\sum_{p \in \mathcal{P}} \mathbb{I}(C \in \mathcal{B}_p)$ which will depend only on the length of the cycle, and naturally the total number of nodes (observations) $n$. Denoting this frequency by $f(n,j)$, this can be determined via a straightforward counting argument.

Consider 2-cycles ($j=2$). Since there are $n!$ possible bases in total and every possible 2-cycle is extracted in every basis (as we take every possible pair of nodes and close the loop using a back-edge), the frequency of any 2-cycle is simply $f(n, 2) = n!$.

For cycles of length $j \ge 3$, the construction algorithm extracts a cycle if and only if its $j$ vertices appear consecutively in the H-path, forming a forward segment that is closed by a single back-edge. A specific $j$-cycle can be formed by $j$ distinct rotations of its vertices (e.g., a 3-cycle $A \to B \to C \to A$ can be formed by the sequence $A \to B \to C$, or $B \to C \to A$, or $C \to A \to B$). For any one of these $j$ specific sequences, we can treat those $j$ vertices as a single indivisible block. When constructing a path of length $n$, we are essentially placing this 1 block alongside the remaining $n-j$ individual vertices. For example, if $n=5$ and we are placing the specific sequence $[A \to B \to C]$, we are ordering $1$ block and $2$ remaining vertices (say, $D$ and $E$). This gives us a total of $n-j+1$ discrete objects to order (the block $[A \to B \to C]$, $D$, and $E$). Since the remaining $n-j$ vertices can be arranged in $(n-j)!$ ways, exactly $(n-j+1) \times (n-j)! = (n-j+1)!$ paths will contain this specific sequence. Summing over all $j$ possible rotations yields $f(n, j) = j(n-j+1)! \text{ for } j \ge 3$.

Given these frequencies, for any generic basis $\mathcal{B}_p$ generated by an H-path $p \in \mathcal{P}$, let $V(\mathcal{B}_p)$ denote the set of strictly GARP-violating cycles within that basis. Generalizing from the example, we define the frequency-adjusted sum of the violating cycle costs $M(p)$, and the corresponding frequency-adjusted weight $W(p)$, as:
\begin{equation*}
    M(p) = \sum_{C \in V(\mathcal{B}_p)} \frac{\text{MPI}_D(C)}{f(n, |C|)}, \quad W(p) = \sum_{C \in V(\mathcal{B}_p)} \frac{1}{f(n, |C|)}.
\end{equation*}

We can now define the BA-MPI for a generic dataset $D$ as:
\begin{equation*}
    \text{BA-MPI}_D = \frac{ \sum_{p \in \mathcal{P}} M(p) }{ \sum_{p \in \mathcal{P}} W(p) }.
\end{equation*}

This allows us to state the equivalence between the original MPI, $\text{OA-MPI}_D = \frac{1}{|V|} \sum_{C \in V} \text{MPI}_D(C)$, and the $\text{BA-MPI}_D$:

\begin{theorem} \label{thm:ba_equiv}
For any GARP-violating dataset $D$, $\text{BA-MPI}_{D} = \text{OA-MPI}_{D}$.
\end{theorem}

Theorem \ref{thm:ba_equiv} formally establishes that estimating the BA-MPI is identical to estimating the OA-MPI, using all possible H-paths, rather than the enumeration of the entire cycle space. 

Directly enumerating all cycles to compute the exact OA-MPI is extremely costly, requiring $\mathcal{O}(n!)$ time to evaluate with \citet{smeulders2013note} proving that calculating the OA-MPI is an NP-hard problem. As Theorem \ref{thm:ba_equiv} establishes that the exact BA-MPI and exact OA-MPI yield identical values for any input dataset, computing the exact BA-MPI is also NP-hard. In fact, directly evaluating the BA-MPI over all $n!$ paths is asymptotically worse, taking $\mathcal{O}(n^2 n!)$ operations, since we must extract and evaluate the cycle basis for each H-path. However, this leads us to a natural estimator of the BA-MPI. 

\textbf{BA-MPI Estimator:} Despite Theorem \ref{thm:ba_equiv} guaranteeing an exact equivalence between the BA-MPI and the OA-MPI, computing all $n!$ paths remains intractable (e.g., a 30-observation dataset yields $30! \approx 2.65 \times 10^{32}$ paths). However, the BA-MPI has essentially transformed a combinatorial graph problem into a statistical problem, allowing us to use a random sample from the $n!$ H-paths as an estimator for the true BA-MPI, and thus the OA-MPI.

I define the \textit{BA-MPI Estimator} ($\text{BA-MPI}_{D,K}$) as the sample ratio estimator. To compute this, we draw a random sequence of $K$ independent H-paths $(p_1, \dots, p_K)$ uniformly with replacement from $\mathcal{P}$, inducing a sequence of $K$ cycle bases $(\mathcal{B}_{p_1}, \dots, \mathcal{B}_{p_K})$, and evaluate the sample ratio:
\begin{equation}
    \text{BA-MPI}_{D,K} = \frac{ \frac{1}{K} \sum_{i=1}^K M(p_i) }{ \frac{1}{K} \sum_{i=1}^K W(p_i) }.
\end{equation}

\begin{proposition} \label{prop:mean_MPI}
For any GARP-violating dataset $D$, the BA-MPI estimator ($\text{BA-MPI}_{D,K}$) is a strongly consistent estimator of the OA-MPI ($\text{OA-MPI}_{D}$) as the number of sampled paths $K \to \infty$:
\begin{equation}
    \text{BA-MPI}_{D,K} \xrightarrow{\text{a.s.}} \text{OA-MPI}_{D}.
\end{equation}
\end{proposition}

If we consider $\text{BA-MPI}_{D}$ as an estimand which is equivalent to $\text{OA-MPI}_{D}$ by Theorem \ref{thm:ba_equiv}, it is easy to see that Proposition \ref{prop:mean_MPI} holds using the Strong Law of Large Numbers and the Continuous Mapping Theorem.

Crucially, the BA-MPI estimator overcomes the combinatorial bottleneck by drawing a finite sample of $K$ H-paths. For each sampled path $v_1 \to v_2 \to \dots \to v_n$, we pre-compute three cumulative prefix arrays along the path in $\mathcal{O}(n)$ time: cumulative surplus $P^{\text{cost}}_t = \sum_{s=1}^{t-1} s_{v_s, v_{s+1}}$, cumulative expenditure $P^{\text{exp}}_t = \sum_{s=1}^{t} E_{v_s}$, and cumulative count of unaffordable edges $P^{\text{bad}}_t = \sum_{s=1}^{t-1} \mathbb{I}(s_{v_s, v_{s+1}} < 0)$. For any candidate contiguous cycle closed by back-edge $v_j \to v_i$ ($j > i$), the forward path subsegment satisfies affordability on all edges if and only if $P^{\text{bad}}_j - P^{\text{bad}}_i = 0$, verified in $\mathcal{O}(1)$ time. If there are no unaffordable edges, the total cycle surplus $(P^{\text{cost}}_j - P^{\text{cost}}_i) + s_{v_j, v_i}$ and total cycle expenditure $P^{\text{exp}}_j - P^{\text{exp}}_{i-1}$ (with $P^{\text{exp}}_0 = 0$) are both obtained in $\mathcal{O}(1)$ time. As each basis contains $(n-1)^2$ cycles, verifying validity and evaluating the MPI across all cycles in the basis requires strictly $\mathcal{O}(n^2)$ operations. With the computational cost of drawing a random path permutation being only $\mathcal{O}(n)$, each sampled basis is processed in $\mathcal{O}(n^2)$ time. Thus, by choosing $K \ll n!$, the total computational complexity to evaluate $K$ randomly sampled bases is strictly polynomial at $\mathcal{O}(K n^2)$, providing an efficient, consistent estimator for the OA-MPI, whose direct exact computation is fundamentally NP-hard.\footnote{Note that this $\mathcal{O}(Kn^2)$ runtime does not contradict the NP-hardness result of \citet{smeulders2013note}. For any chosen number of draws $K$, evaluating the $K$ bases takes $\mathcal{O}(Kn^2)$ time, and for any fixed dataset $D$, the estimator converges to the true OA-MPI as $K \to \infty$ (Proposition \ref{prop:mean_MPI}). In a hypothetical worst-case dataset where GARP violations occur only in very long cycles, a random H-path would rarely extract a violating cycle, so an intractably large $K$ would be needed for an accurate estimate. In typical consumer data, however, violations frequently occur in shorter cycles, allowing the estimator to converge rapidly with a modest $K$.}

Before proceeding, we discuss why sampling H-paths to extract cycle bases is vastly superior to drawing individual cycles to form a Monte Carlo estimator. In principle, one could estimate the mean MPI by uniformly sampling cycles directly from the graph. This, however, is typically intractable because it implicitly requires counting/enumerating the total number of cycles.\footnote{Specifically, counting or uniformly sampling simple cycles in a directed graph without enumerating the entire cycle space is a \#P-complete problem \citep[Section 4, Problem 14]{valiant1979complexity}.} Alternatively, sampling from the complete cycle space and applying rejection sampling is computationally unviable, as the proportion of violating cycles is typically tiny relative to the size of the cycle space. In contrast, drawing an H-path uniformly at random is computationally trivial. By sampling paths instead of individual cycles, we circumvent these sampling bottlenecks, rapidly generating structured batches of $(n-1)^2$ cycles with exactly known frequencies.

While Proposition \ref{prop:mean_MPI} establishes consistency, as a ratio estimator, clearly the BA-MPI estimator will exhibit some form of finite-sample bias.\footnote{In the sense that the expectation of a ratio is, in general, not equal to the ratio of the expectations.} In the following proposition, we establish an upper bound on this bias. Let $\mu_M = \mathbb{E}[M(p)]$ and $\mu_W = \mathbb{E}[W(p)]$ denote the true population expectations of the basis sum and basis weight, respectively, with $\sigma_M^2 = \text{Var}[M(p)]$ and $\sigma_W^2 = \text{Var}[W(p)]$ denoting their respective population variances. 

\begin{proposition} \label{prop:mean_bias}
Let $W_{\min} = \min_{p \in \mathcal{P}} W(p) > 0$ denote the population minimum basis weight in dataset $D$, and let $R(p) = M(p) - \theta W(p)$ denote the basis residual relative to the true ratio $\theta = \mu_M / \mu_W = \text{OA-MPI}_D$, with variance $\sigma_R^2 = \text{Var}[R(p)]$. The absolute finite-sample bias of the BA-MPI estimator is bounded by:
\begin{equation}
    \left| \mathbb{E}[\text{BA-MPI}_{D,K}] - \text{OA-MPI}_D \right| \le \min\left( 1, \; \frac{1}{K} \left[ \frac{\sigma_R \sigma_W}{\mu_W W_{\min}} \right] \right).
\end{equation}
\end{proposition}

Naturally, this upper bound on the bias shrinks proportionally to the reciprocal of the number of draws, $K$. A strictly positive minimum basis weight $W_{\min} > 0$ is guaranteed whenever the dataset contains at least one 2-cycle violation, as every 2-cycle is extracted in the basis of every H-path. In typical revealed preference datasets, 2-cycle violations are ubiquitous. If a dataset contains only higher-order cycles, conditioning sampling on bases with non-empty violating cycle sets still guarantees $W_{\min}^+ > 0$ and preserves the exact ratio estimand (see Proof of Proposition \ref{prop:mean_bias}). As an analytical upper bound without distributional assumptions, this bound is typically conservative; in empirical applications, the reported bounds are the sample analogues computed from the $K$ draws. 

\subsection{Basis Percentile MPI}
Analogously, we can also compute the median MPI of a dataset, or indeed, any arbitrary $\alpha$-percentile. We refer to this as the \textit{Basis $\alpha$-Percentile MPI} ($\text{BP}_\alpha\text{-MPI}_D$). The $\alpha$-percentile MPI originally proposed by \citet{echenique2011money} is denoted by $\text{OP}_\alpha\text{-MPI}_D$ (Original $\alpha$-Percentile MPI).

Recall from our previous example of a 3-observation dataset $D$, the global set of strictly violating cycles contained exactly three elements with MPIs of $0.125$, $0.200$, and $0.250$. The $\text{OP}_{50}\text{-MPI}_D$ is simply the median of $0.125$, $0.200$, and $0.250$, which is $0.200$. Naturally, this extends to any arbitrary $\alpha$-percentile. For instance, the $\text{OP}_{25}\text{-MPI}_D$ is $0.125$.\footnote{As the $\alpha$-percentile is formally defined as the smallest value where the CDF reaches or exceeds $\alpha$, the 25th percentile is $0.125$.}

Formally, for any global set of strictly violating cycles $V$, we define the ``original'' cumulative distribution function ($\text{O-CDF}_D$) as the exact proportion of all cycles with an MPI less than or equal to $x$:
\begin{equation*}
    \text{O-CDF}_D(x) = \frac{1}{|V|} \sum_{C \in V} \mathbb{I}(\text{MPI}_D(C) \le x).
\end{equation*}
The $\text{OP}_\alpha\text{-MPI}_D$ is the smallest value $M_\alpha$ such that $\text{O-CDF}_D(M_\alpha) \ge \alpha$. 

To find the equivalent CDF using the weakly fundamental directed cycle bases, we use a weighted and renormalized CDF, denoted $\text{B-CDF}_D(x)$. For every path $p \in \mathcal{P}$, let $V(\mathcal{B}_p)$ denote the set of strictly violating cycles extracted from its corresponding fundamental cycle basis $\mathcal{B}_p$. With the same weights, $w(C) = 1/f(n, |C|)$, the $\text{B-CDF}_D(x)$ over all paths is given by:
\begin{equation*}
    \text{B-CDF}_D(x) = \frac{ \sum_{p \in \mathcal{P}} \sum_{C \in V(\mathcal{B}_p)} w(C) \mathbb{I}(\text{MPI}_D(C) \le x) }{ \sum_{p \in \mathcal{P}} \sum_{C \in V(\mathcal{B}_p)} w(C) }.
\end{equation*}

The $\text{OP}_\alpha\text{-MPI}_D$ is then calculated from the $\text{O-CDF}_D$, and the $\text{BP}_\alpha\text{-MPI}_D$ is defined as the specific MPI value $M_\alpha$ where $\text{B-CDF}_D(M_\alpha) \ge \alpha$. 

Continuing with our three-observation example, the cycle weights are $w(C_{12}) = 1/6$, $w(C_{23}) = 1/6$, and $w(C_{123}) = 1/3$. As in the $\text{BA-MPI}_D$ calculation, the denominator of $\text{B-CDF}_D(x)$ is $3$. At $x = 0.15$, only $C_{23}$ contributes mass across all paths, yielding:
\begin{equation*}
    \text{B-CDF}_D(0.15) = \frac{6(1/6)}{3} = \frac{1}{3}.
\end{equation*}
At $x = 0.20$, the remaining 3-cycle enters, yielding:
\begin{equation*}
    \text{B-CDF}_D(0.20) = \frac{6(1/6) + 3(1/3)}{3} = \frac{2}{3} \approx 0.667.
\end{equation*}
Finally, at $x = 0.250$, cycle $C_{12}$ enters and the CDF reaches unity ($\text{B-CDF}_D(0.250) = 1.0$).
Thus, to find the $\text{BP}_{50}\text{-MPI}_D$, we see that the value where the $\text{B-CDF}_D$ first reaches or exceeds $0.50$ is exactly $0.20$, matching the $\text{OP}_{50}\text{-MPI}_D$ value. Similarly, for the 25th percentile, the $\text{B-CDF}_D$ reaches or exceeds $0.25$ at the exact support point $x = 0.125$ ($\text{BP}_{25}\text{-MPI}_D = 0.125$), mirroring $\text{OP}_{25}\text{-MPI}_D$.

We can now state an equivalent to Theorem \ref{thm:ba_equiv} for the $\text{BP}_\alpha\text{-MPI}_D$.

\begin{theorem} \label{thm:bp_equiv}
For any GARP-violating dataset $D$ and any $\alpha \in (0, 1)$, $\text{BP}_\alpha\text{-MPI}_D = \text{OP}_\alpha\text{-MPI}_D$.
\end{theorem}

In practice, to compute this exact percentile without enumerating all cycles, we evaluate the empirical cumulative distribution function over the same sequence of $K$ sampled cycle bases $(\mathcal{B}_{p_1}, \dots, \mathcal{B}_{p_K})$. The estimator for $\text{B-CDF}_D(x)$ is given by:
\begin{equation}
    \text{B-CDF}_{D,K}(x) = \frac{\frac{1}{K} \sum_{i=1}^K \sum_{C \in V(\mathcal{B}_{p_i})} w(C) \mathbb{I}(\text{MPI}_D(C) \le x) }{ \frac{1}{K} \sum_{i=1}^K \sum_{C \in V(\mathcal{B}_{p_i})} w(C) } = \frac{\bar{Y}_K(x)}{\bar{W}_K}.
\end{equation}
By explicitly tracking this percentile across the sampled cycle bases, we get precise approximations of the true percentiles. The following propositions make this explicit in terms of consistency, and the exponential decay of the bias.

\begin{proposition} \label{prop:percentile_MPI}
For any GARP-violating dataset $D$ and any $\alpha \in (0, 1)$, let $\mathcal{X}$ denote the finite support of cycle MPIs, and let $M_\alpha^- = \min \{x \in \mathcal{X} : \text{B-CDF}_D(x) \ge \alpha\}$ and $M_\alpha^+ = \min \{x \in \mathcal{X} : \text{B-CDF}_D(x) > \alpha\}$.
\begin{enumerate}
    \item[(i)] If $\text{B-CDF}_D(M_\alpha^-) > \alpha$, so that $M_\alpha^- = M_\alpha^+ = M_\alpha$, then:
    \begin{equation}
        \text{BP}_\alpha\text{-MPI}_{D,K} \xrightarrow{\text{a.s.}} M_\alpha \quad \text{as } K \to \infty.
    \end{equation}
    \item[(ii)] If $\text{B-CDF}_D(M_\alpha^-) = \alpha$, so that $M_\alpha^- < M_\alpha^+$, then almost surely:
    \begin{equation}
        \text{BP}_\alpha\text{-MPI}_{D,K} \in \{M_\alpha^-, M_\alpha^+\} \quad \text{for all sufficiently large } K.
    \end{equation}
\end{enumerate}
\end{proposition}

Proposition \ref{prop:percentile_MPI} distinguishes between two cases because $\text{B-CDF}_D(x)$ is a discrete step function and is therefore not injective (one-to-one) along its flat horizontal steps. In part (ii) of the proposition, where $\text{B-CDF}_D(M_\alpha^-) = \alpha$, the threshold $\alpha$ coincides with a flat horizontal step of the CDF between the two consecutive support points $M_\alpha^-$ and $M_\alpha^+$, so for sufficiently large $K$ the sample percentile lands on one of these two step endpoints, $\{M_\alpha^-, M_\alpha^+\}$. By contrast, in part (i) of the proposition, where $\text{B-CDF}_D(M_\alpha^-) > \alpha$, the threshold $\alpha$ crosses the CDF at a vertical jump, so the true $\alpha$-percentile is uniquely defined ($M_\alpha^- = M_\alpha^+ = M_\alpha$) and is separated from the adjacent flat steps by a strictly positive vertical margin $\delta > 0$:
\begin{equation*}
    \delta = \min \left( \alpha - \max_{x < M_\alpha} \text{B-CDF}_D(x), \; \text{B-CDF}_D(M_\alpha) - \alpha \right) > 0.
\end{equation*}
Any estimation error in the empirical CDF strictly smaller than $\delta$ guarantees that the empirical step function crosses the horizontal threshold $\alpha$ at the exact true percentile $M_\alpha$.\footnote{Following the convention that $\max \emptyset = 0$ when $M_\alpha = \min \mathcal{X}$.} Analogous to Proposition \ref{prop:mean_bias}, we can also bound the finite-sample bias of this Basis $\alpha$-Percentile estimator:

\begin{proposition} \label{prop:percentile_bias}
For any GARP-violating dataset $D$ and any $\alpha \in (0, 1)$ such that the $\alpha$-quantile is uniquely isolated ($\text{B-CDF}_D(M_\alpha) > \alpha$, so that $\delta > 0$), let $W_{\min} = \min_{p \in \mathcal{P}} W(p) > 0$ and $W_{\max} = \max_{p \in \mathcal{P}} W(p)$. The absolute finite-sample bias of the Basis $\alpha$-Percentile estimator is bounded by:
\begin{equation}
    \left| \mathbb{E}[\text{BP}_\alpha\text{-MPI}_{D,K}] - M_\alpha \right| \le \min\left( 1, \; 4 \exp\left( - \frac{2 K \delta^2 W_{\min}^2}{W_{\max}^2} \right) \right).
\end{equation}
\end{proposition}

As in Proposition \ref{prop:mean_bias}, a strictly positive minimum basis weight $W_{\min} > 0$ holds whenever the dataset contains at least one 2-cycle violation (or under conditional sampling on non-empty bases). As the true population of cycle MPIs is discrete, when the quantile is uniquely isolated, the strictly positive vertical gap $\delta$ isolates the true percentile $M_\alpha$. Bias in the estimate for the percentile arises when the empirical CDF deviates from the true CDF by $\delta$ or more at critical points, causing the approximated CDF to cross the $\alpha$ threshold at the wrong step. The theoretical bound in Proposition \ref{prop:percentile_bias} provides a distribution-free worst-case guarantee that decays exponentially as the sample size $K$ grows. As it is a distribution-free bound evaluated over the extreme weights $W_{\min}$ and $W_{\max}$, it is conservative in finite samples. 

Like the exact OA-MPI, computing any $\text{OP}_\alpha\text{-MPI}_D$ directly is an NP-hard problem \citep{smeulders2013note}. I overcome this computational issue using the consistent estimator for the B-CDF, which allows researchers to construct the empirical CDF in polynomial time $\mathcal{O}(K n^2)$.

\section{Empirical Application} \label{sec:empirical}

To validate the proposed estimators in practice, I replicate the empirical analysis of \citet{echenique2011money} and \citet{smeulders2013note} using scanner data from the Stanford Basket Dataset. The dataset tracks the grocery purchases of 494 households over 26 weeks across 375 distinct goods with 396 households strictly violating GARP. For each GARP-violating household, I compute: (i) the minimum and maximum MPI bounds of \citet{smeulders2013note}; (ii) the approximated average and median MPI (by focusing on cycles of lengths 2, 3, and 4) of \citet{echenique2011money}; and (iii) the basis average and median MPI estimators ($\text{BA-MPI}_{D,K}$ and $\text{BP}_{50}\text{-MPI}_{D,K}$). I also compute the upper bounds on the finite-sample bias of these estimators from Proposition \ref{prop:mean_bias} and Proposition \ref{prop:percentile_bias} replacing population parameters with sample analogues. Additionally, recognizing that these analytical bounds are conservative worst-case scenarios, I quantify finite-sample bias using non-parametric bootstrap by drawing 1,000 bootstrap resamples (with replacement) from the set of sampled H-paths.\footnote{To optimize computational execution, the empirical algorithm implements an automated convergence check over valid (non-empty) bases $K$: after a burn-in of $1{,}000$ valid bases, the sampler halts if the running ratio changes by less than $10^{-3}$ for $200$ consecutive valid bases (capped at $5{,}000$ total path draws). For the 395 households containing at least one 2-cycle, every sampled H-path yields a non-empty basis, and 354 out of 396 violating households (89.4\%) satisfy convergence immediately at $K = 1{,}200$ draws (with an overall average of $1{,}228.85$ sampled H-paths across all 396 households, including the single household lacking 2-cycles that reaches the $5{,}000$-draw cap).} 

Table \ref{tab:mpi_stats} reports the descriptive statistics of these indices across all GARP-violating households. Typically, revealed preference violations are heavily concentrated in shorter cycles (predominantly 2-cycles and 3-cycles). As shorter cycles dominate the total count of violations, the approximated mean MPI of \citet{echenique2011money} (which truncates at $k \le 4$) aligns closely with the Basis Average MPI ($0.0610$ vs $0.0607$, a difference of merely three basis points or $0.0003$), as also shown in Figure \ref{fig:cdf_plots} (left).\footnote{The observation that longer cycles provide largely redundant severity information aligns with recent theoretical findings by \citet[Lemma 3]{caradonna2026inconsistency}, who formally proves that in sufficiently rich environments, any choice cycle of length greater than 3 necessarily mathematically forces the existence of strictly shorter cycles.}

\begin{table}[ht]
    \centering
    \footnotesize
    \caption{Descriptive Statistics of Money Pump Indices across GARP-violating Households ($N=396$)}
    \label{tab:mpi_stats}
    \vspace{0.5em}
    \begin{minipage}{320pt}
    \setlength{\tabcolsep}{4.5pt}
    \begin{tabular*}{\linewidth}{@{\extracolsep{\fill}}l c c c c c c c}
    \toprule
    & \multicolumn{3}{c}{\textbf{Min/Max MPI}} & \multicolumn{2}{c}{\textbf{Approximate MPI}} & \multicolumn{2}{c}{\textbf{Basis MPI}} \\
    \cmidrule(lr){2-4} \cmidrule(lr){5-6} \cmidrule(lr){7-8}
    \textbf{Statistic} & Min & Max & Range & Average & Median & Average & Median \\
    \midrule
    Average & 0.0341 & 0.0936 & 0.0595 & 0.0610 & 0.0592 & 0.0607 & 0.0587 \\
    Standard deviation & 0.0288 & 0.0616 & 0.0604 & 0.0358 & 0.0369 & 0.0356 & 0.0368 \\
    Number of zeros & - & - & 75 & - & - & - & - \\
    Minimum & 0.0002 & 0.0048 & 0.0000 & 0.0048 & 0.0048 & 0.0048 & 0.0048 \\
    First quartile & 0.0154 & 0.0489 & 0.0116 & 0.0356 & 0.0326 & 0.0354 & 0.0333 \\
    Median & 0.0268 & 0.0797 & 0.0443 & 0.0551 & 0.0497 & 0.0548 & 0.0493 \\
    Third quartile & 0.0425 & 0.1274 & 0.0863 & 0.0809 & 0.0789 & 0.0796 & 0.0763 \\
    Maximum & 0.2782 & 0.4010 & 0.3350 & 0.2782 & 0.2782 & 0.2782 & 0.2782 \\
    \bottomrule
\end{tabular*}

    \end{minipage}
\end{table}

A parallel pattern emerges for the median MPI. Across all 396 GARP-violating households, the Basis Median MPI averages $0.0587$ (standard deviation $0.0368$), differing from the Approximated Median MPI of $0.0592$ (standard deviation $0.0369$) by just $0.0005$, with the two distributions tracking each other closely in Figure \ref{fig:cdf_plots} (right). Together, these results confirm that for short-panel scanner datasets like the Stanford Basket panel, truncating at cycle lengths up to 4 as in \citet{echenique2011money} introduces virtually no bias. At the same time, for datasets with a larger number of observations---where enumerating cycles up to length 4 becomes computationally costly and risks omitting longer violating cycles---our Basis MPI estimators remain fast to compute in $\mathcal{O}(Kn^2)$ polynomial time while accounting for cycles of all lengths.

Visually, the basis mean and median MPI distributions fit symmetrically between the minimum and maximum bounds of \citet{smeulders2013note}, as displayed in Figure \ref{fig:cdf_plots}. While these bounds provide tight sandwiches for households with simple violation structures (and coincide exactly for 75 households), their wide average range ($0.0595$ on average) underscores the value of the basis estimator in pinning down point estimates of irrationality when bounds remain uninformative. 

\begin{figure}[ht]
    \centering
    \includegraphics[width=0.48\textwidth]{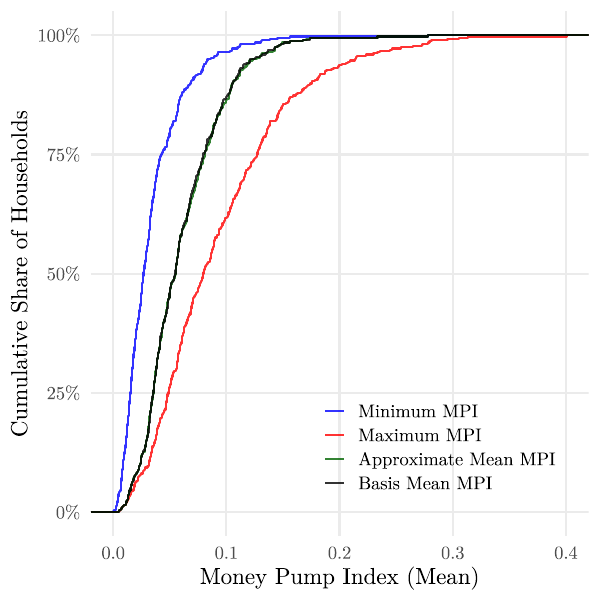} \hfill
    \includegraphics[width=0.48\textwidth]{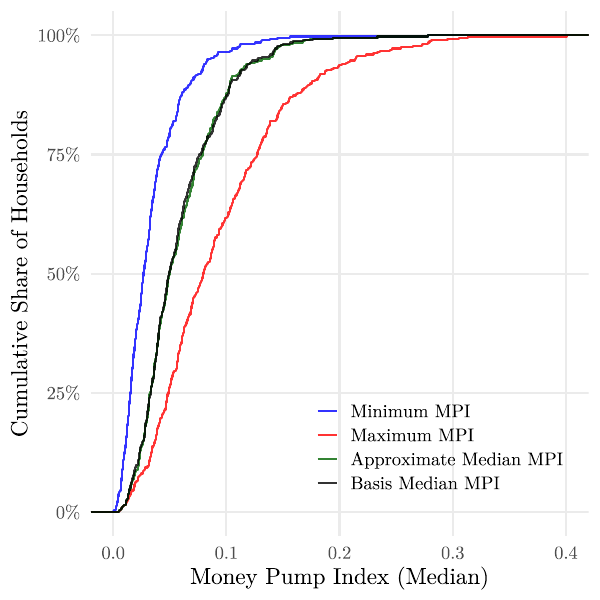}
    \caption{\textbf{Empirical Cumulative Distribution Functions of Money Pump Indices across Households ($N = 396$).} Left panel: Basis Mean MPI versus Approximated Mean MPI ($k \le 4$). Right panel: Basis Median MPI, Approximated Median MPI ($k \le 4$), Minimum MPI, and Maximum MPI. The cumulative share represents the proportion of violating households with MPI less than or equal to $x$.} \label{fig:cdf_plots}
\end{figure}

Table \ref{tab:bias_bounds} reports the finite-sample empirical bias and theoretical bounds of the basis estimators across all households. For the mean MPI, the theoretical upper bounds on bias are tight for the vast majority of households (mean $0.0011$, median $0.0001$, with 161 households achieving zero bound). In contrast, the theoretical upper bound on the median MPI bias is conservative (with 94 households achieving zero bound), reflecting that it is a distribution-free worst-case guarantee whose analytical bound naturally caps at $1$ when basis weights exhibit wide dispersion. From Figure \ref{fig:bias_bounds_hist}, the non-parametric bootstrap empirical bias provides the operational verification: both the average empirical bias for the mean MPI and the median MPI are tightly clustered around zero, with sample medians of 0 and exact zeros in 161 and 170 households, respectively. This demonstrates that in practice, the basis estimators achieve near-zero finite-sample bias with modest draws.

\begin{table}[ht]
    \centering
    \footnotesize
    \caption{Finite-Sample Bias and Theoretical Bounds of Basis MPI Estimators}
    \label{tab:bias_bounds}
    \vspace{0.5em}
    \begin{minipage}{270pt}
    \setlength{\tabcolsep}{10pt}
    \begin{tabular*}{\linewidth}{@{\extracolsep{\fill}}l c c c c}
    \toprule
    & \multicolumn{2}{c}{\textbf{Empirical Bias}} & \multicolumn{2}{c}{\textbf{Theoretical Bounds}} \\
    \cmidrule(lr){2-3} \cmidrule(lr){4-5}
    \textbf{Statistic} & Average & Median & Average & Median \\
    \midrule
    Average & 0.0001 & 0.0000 & 0.0011 & 0.7603 \\
    Standard deviation & 0.0007 & 0.0037 & 0.0043 & 0.4273 \\
    Number of zeros & 161 & 170 & 161 & 94 \\
    Minimum & -0.0017 & -0.0170 & 0.0000 & 0.0000 \\
    First quartile & 0.0000 & -0.0001 & 0.0000 & 1.0000 \\
    Median & 0.0000 & 0.0000 & 0.0001 & 1.0000 \\
    Third quartile & 0.0001 & 0.0001 & 0.0004 & 1.0000 \\
    Maximum & 0.0101 & 0.0288 & 0.0477 & 1.0000 \\
    \bottomrule
\end{tabular*}

    \end{minipage}
\end{table}

\begin{figure}[ht]
    \centering
    \includegraphics[width=1\textwidth]{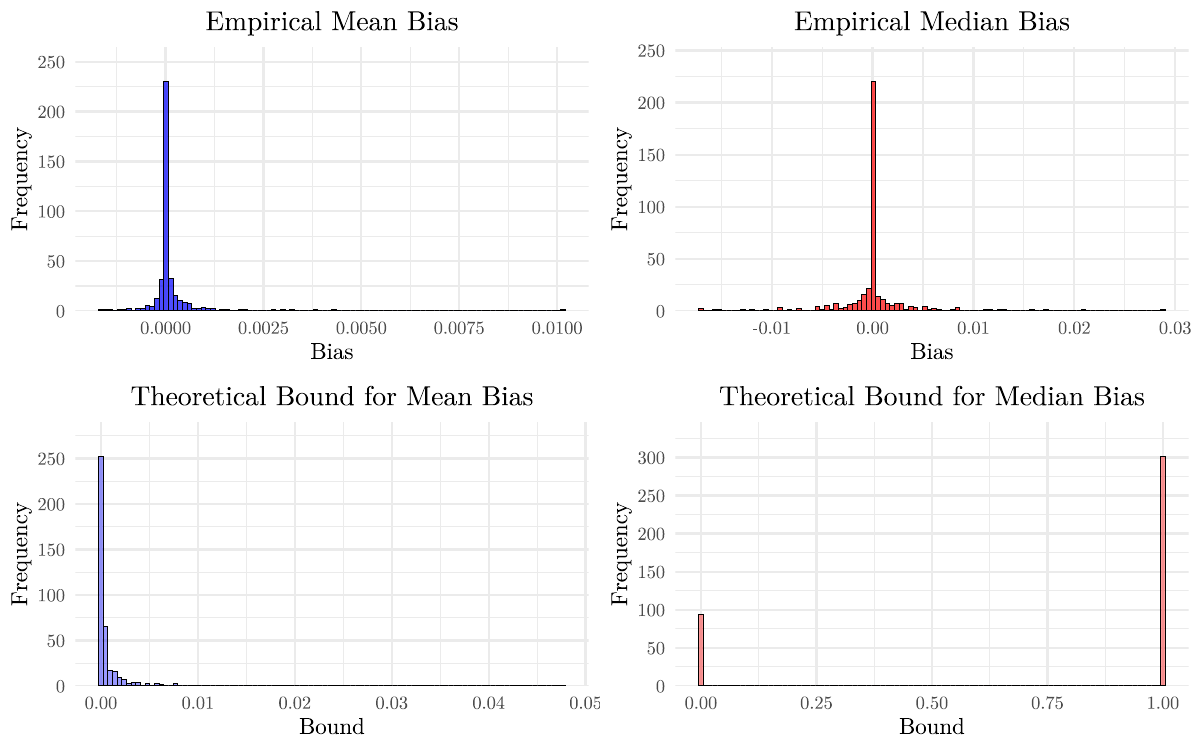}
    \caption{Histograms of Finite-Sample Empirical Bias and Analytical Theoretical Bounds across all $N = 396$ GARP-violating Households. The top row displays the empirical bias evaluated via 1,000 bootstrap resamples for the Basis Average (left) and Basis Median (right) MPI estimators. The bottom row displays the corresponding analytical finite-sample upper bounds from Proposition \ref{prop:mean_bias} (left) and Proposition \ref{prop:percentile_bias} (right).}
    \label{fig:bias_bounds_hist}
\end{figure}

For 75 households, the minimum and maximum MPI bounds coincide; hence, this bound represents their exact $\text{OA-MPI}_D$ \citep{smeulders2013note}. When computing the basis estimators for these 75 households, both the Basis Average and Basis Median match the exact Min/Max MPI with zero discrepancy providing direct empirical validation of the proposed estimators against the true values.

\FloatBarrier
\section{Conclusion} \label{sec:conclusion}

This paper proposes a computationally tractable approach to calculate the mean and median Money Pump Index (MPI) of \citet{echenique2011money}. Without relying on shorter cycle approximations or min-max bounds to avoid NP-hard complexity issues, this approach conceptualizes choice violations as a directed graph of expenditures. The violations are then projected onto the graph's weakly fundamental directed cycle bases through the random sampling of Hamiltonian paths from which the bases are derived. Theoretical results establish that these estimators are consistent and computable in polynomial time. When applied to the Stanford Basket Dataset, it is shown that any finite-sample bias is negligible. Additionally, the basis mean and median MPI estimators demonstrate remarkable agreement with short-cycle approximations. Finally, where the minimum and maximum MPI bounds are the same, the basis estimators also coincide. By reducing the computational burden, this methodology allows researchers to accurately estimate the true MPI tractably.

\begin{appendices}
\section{Proofs} \label{sec:appendix}

\noindent\textbf{Proof of Theorem \ref{thm:ba_equiv}.}

Let $V$ be the finite set of all strictly GARP-violating cycles within the complete directed expenditure graph. By definition, for any dataset exhibiting at least one GARP violation ($|V| \ge 1$):

\begin{equation*}
    \text{OA-MPI}_D = \frac{1}{|V|} \sum_{C \in V} \text{MPI}_D(C).
\end{equation*}

Let $\mathcal{P}$ denote the set of all $n!$ directed Hamiltonian paths on the complete graph of observations. Each path $p \in \mathcal{P}$ deterministically induces a weakly fundamental directed cycle basis $\mathcal{B}_p$. Let $V(\mathcal{B}_p) = \mathcal{B}_p \cap V$ denote the subset of generated cycles in $\mathcal{B}_p$ that strictly violate GARP.

Recall the definition of the Basis Average MPI across the ensemble of path-indexed bases:

\begin{equation*} 
    \text{BA-MPI}_D = \frac{ \sum_{p \in \mathcal{P}} \sum_{C \in V(\mathcal{B}_p)} \frac{\text{MPI}_D(C)}{f(n, |C|)} }{ \sum_{p \in \mathcal{P}} \sum_{C \in V(\mathcal{B}_p)} \frac{1}{f(n, |C|)} }.
\end{equation*}

We first consider the numerator of $\text{BA-MPI}_D$. Let $\mathbb{I}(C \in \mathcal{B}_p)$ be an indicator function that equals $1$ if cycle $C$ is in the basis generated by path $p$, and $0$ otherwise. This allows us to rewrite the numerator as:
\begin{equation*}
    \sum_{p \in \mathcal{P}} \sum_{C \in V(\mathcal{B}_p)} \frac{\text{MPI}_D(C)}{f(n, |C|)} = \sum_{p \in \mathcal{P}} \sum_{C \in V} \mathbb{I}(C \in \mathcal{B}_p) \frac{\text{MPI}_D(C)}{f(n, |C|)}.
\end{equation*}
As the limits of the inner sum are now independent of the path index $p$, we can rearrange such that:
\begin{align}
    \sum_{p \in \mathcal{P}} \left( \sum_{C \in V} \mathbb{I}(C \in \mathcal{B}_p) \frac{\text{MPI}_D(C)}{f(n, |C|)} \right) 
    &= \sum_{C \in V} \left( \sum_{p \in \mathcal{P}} \mathbb{I}(C \in \mathcal{B}_p) \frac{\text{MPI}_D(C)}{f(n, |C|)} \right) \nonumber \\[0.85em]
    &= \sum_{C \in V} \left( \sum_{p \in \mathcal{P}} \mathbb{I}(C \in \mathcal{B}_p) \right) \frac{\text{MPI}_D(C)}{f(n, |C|)} \nonumber \\[0.85em]
    &= \sum_{C \in V} \text{MPI}_D(C). \label{eq:swapped_num}
\end{align}

Notice that Equation \eqref{eq:swapped_num} depends on $\sum_{p \in \mathcal{P}} \mathbb{I}(C \in \mathcal{B}_p)$, which is the marginal inclusion frequency of cycle $C$ across the ensemble of generated bases. By construction of the weakly fundamental basis cycles, each cycle of length $|C|$ appears in exactly $f(n, |C|)$ path-indexed bases, meaning $\sum_{p \in \mathcal{P}} \mathbb{I}(C \in \mathcal{B}_p) = f(n, |C|)$.

Applying the same summation swap to the denominator yields:
\begin{align*}
    \sum_{p \in \mathcal{P}} \sum_{C \in V(\mathcal{B}_p)} \frac{1}{f(n, |C|)} 
    &= \sum_{p \in \mathcal{P}} \sum_{C \in V} \mathbb{I}(C \in \mathcal{B}_p) \frac{1}{f(n, |C|)} \\[0.85em]
    &= \sum_{C \in V} \left( \sum_{p \in \mathcal{P}} \mathbb{I}(C \in \mathcal{B}_p) \right) \frac{1}{f(n, |C|)} \\[0.85em]
    &= \sum_{C \in V} 1 = |V|.
\end{align*}

Dividing the reduced numerator by the reduced denominator completes the proof. \qed

\vspace{1em}\noindent\textbf{Proof of Proposition \ref{prop:mean_MPI}.}

For the $i$-th drawn path $p_i$, let $\mathcal{B}_{p_i}$ denote the induced weakly fundamental directed cycle basis, and let $V(\mathcal{B}_{p_i}) = \mathcal{B}_{p_i} \cap V$ denote the subset of basis cycles that strictly violate GARP. Let the BA-MPI estimator over $K$ sampled paths be defined as:
\begin{equation*}
    \text{BA-MPI}_{D,K} = \frac{\bar{M}_K}{\bar{W}_K},
\end{equation*}
with sample means $\bar{M}_K$ and $\bar{W}_K$ given by:
\begin{align*}
    \bar{M}_K &= \frac{1}{K} \sum_{i=1}^K \left[ \sum_{C \in V(\mathcal{B}_{p_i})} \frac{\text{MPI}_D(C)}{f(n, |C|)} \right] = \frac{1}{K} \sum_{i=1}^K M(p_i), \\[0.85em]
    \bar{W}_K &= \frac{1}{K} \sum_{i=1}^K \left[ \sum_{C \in V(\mathcal{B}_{p_i})} \frac{1}{f(n, |C|)} \right] = \frac{1}{K} \sum_{i=1}^K W(p_i).
\end{align*}

Since each path $p_i$ is drawn independently and uniformly at random with replacement from $\mathcal{P}$, the sequences $\{M(p_i)\}_{i=1}^K$ and $\{W(p_i)\}_{i=1}^K$ are independent and identically distributed (i.i.d.).

By Theorem \ref{thm:ba_equiv}, the true population expected value of $M(p)$ over uniformly random paths $p$ is:

\begin{equation*}
    \mathbb{E}[M(p)] = \frac{1}{n!} \sum_{p \in \mathcal{P}} M(p) = \frac{1}{n!} \sum_{C \in V} \text{MPI}_D(C).
\end{equation*}

Similarly, the expected value of the basis weight is:

\begin{equation*}
    \mathbb{E}[W(p)] = \frac{1}{n!} \sum_{p \in \mathcal{P}} W(p) = \frac{|V|}{n!}.
\end{equation*}

As $D$ violates GARP, the set of violating cycles is non-empty ($|V| \ge 1$), which guarantees that $\mathbb{E}[W(p)] > 0$.

As the random variables $M(p_i)$ and $W(p_i)$ are i.i.d. and strictly bounded on the finite complete graph of $n$ observations, the Strong Law of Large Numbers (SLLN) implies:

\begin{equation*}
    \bar{M}_K \xrightarrow{\text{a.s.}} \mathbb{E}[M(p)] \quad \text{and} \quad \bar{W}_K \xrightarrow{\text{a.s.}} \mathbb{E}[W(p)] \quad \text{as } K \to \infty.
\end{equation*}

As $\mathbb{E}[W(p)] > 0$, the ratio $g(u, w) = u/w$ is continuous at $(\mathbb{E}[M(p)], \mathbb{E}[W(p)])$. By the Continuous Mapping Theorem (CMT) and Theorem \ref{thm:ba_equiv}:

\begin{align*}
    \text{BA-MPI}_{D,K} = \frac{\bar{M}_K}{\bar{W}_K} \xrightarrow{\text{a.s.}} \frac{\mathbb{E}[M(p)]}{\mathbb{E}[W(p)]} &= \frac{\frac{1}{n!} \sum_{C \in V} \text{MPI}_D(C)}{\frac{|V|}{n!}} \\
    &= \frac{1}{|V|} \sum_{C \in V} \text{MPI}_D(C) = \text{OA-MPI}_{D},
\end{align*}
which establishes strong consistency and completes the proof. \qed

\vspace{1em}\noindent\textbf{Proof of Proposition \ref{prop:mean_bias}.}

Let $\theta = \frac{\mu_M}{\mu_W} = \text{OA-MPI}_D$. For each path $p \in \mathcal{P}$, define the linearized residual:
\begin{equation*}
    R(p) = M(p) - \theta W(p).
\end{equation*}
Notice that $\mathbb{E}[R(p)] = 0$, and let $\sigma_R^2 = \text{Var}[R(p)]$.

Over $K$ i.i.d. sampled paths, let $\bar{R}_K = \frac{1}{K}\sum_{i=1}^K R(p_i) = \bar{M}_K - \theta \bar{W}_K$. The finite-sample estimation error is
\begin{align*}
    \frac{\bar{M}_K}{\bar{W}_K} - \theta 
    &= \frac{\bar{M}_K - \theta \bar{W}_K}{\bar{W}_K} = \frac{\bar{R}_K}{\bar{W}_K} \\[0.85em]
    &= \frac{\bar{R}_K}{\mu_W} - \frac{\bar{R}_K (\bar{W}_K - \mu_W)}{\mu_W \bar{W}_K}.
\end{align*}
Taking expectations on both sides (noting that $\mathbb{E}[\bar{R}_K] = 0$), taking absolute values, and applying Jensen's inequality yields:
\begin{equation*}
     \left| \mathbb{E}[\text{BA-MPI}_{D,K}] - \theta \right| 
    \le \frac{1}{\mu_W}\mathbb{E}\left[ \left| \frac{\bar{R}_K (\bar{W}_K - \mu_W)}{\bar{W}_K} \right| \right].
\end{equation*}

As $W_{\min} > 0$, then $W(p_i) \ge W_{\min} > 0$ for every sampled path, and therefore $\bar{W}_K \ge W_{\min}$. Thus:
\begin{equation*}
    \left| \mathbb{E}[\text{BA-MPI}_{D,K}] - \theta \right| \le \frac{1}{\mu_W W_{\min}} \mathbb{E}\left[ \left| \bar{R}_K (\bar{W}_K - \mu_W) \right| \right].
\end{equation*}

As the $K$ paths are drawn independently, $\text{Var}(\bar{R}_K) = \frac{\sigma_R^2}{K}$ and $\text{Var}(\bar{W}_K) = \frac{\sigma_W^2}{K}$. Applying the Cauchy--Schwarz inequality:
\begin{equation*}
    \mathbb{E}\left[ \left| \bar{R}_K (\bar{W}_K - \mu_W) \right| \right] \le \sqrt{ \text{Var}(\bar{R}_K) \text{Var}(\bar{W}_K) } = \frac{\sigma_R \sigma_W}{K}.
\end{equation*}

Therefore:
\begin{equation*}
    \left| \mathbb{E}[\text{BA-MPI}_{D,K}] - \theta \right| \le \frac{1}{K} \left[ \frac{\sigma_R \sigma_W}{\mu_W W_{\min}} \right].
\end{equation*}

Finally, since both $\text{BA-MPI}_{D,K} \in [0, 1]$ and $\text{OA-MPI}_D \in [0, 1]$, the maximum possible absolute difference cannot exceed $1$, yielding:
\begin{equation*}
    \left| \mathbb{E}[\text{BA-MPI}_{D,K}] - \text{OA-MPI}_D \right| \le \min\left( 1, \; \frac{1}{K} \left[ \frac{\sigma_R \sigma_W}{\mu_W W_{\min}} \right] \right),
\end{equation*}
completing the proof.\footnote{For datasets lacking 2-cycle violations, conditioning sampling on the set of non-empty bases $\mathcal{P}_+ = \{p \in \mathcal{P} : W(p) > 0\}$ yields the identical expected ratio $\frac{\mathbb{E}[M(p) \mid W(p)>0]}{\mathbb{E}[W(p) \mid W(p)>0]} = \text{OA-MPI}_D$ with strictly positive floor $W_{\min}^+ = \min_{p \in \mathcal{P}_+} W(p) > 0$, since $M(p) = W(p) = 0$ on all paths outside $\mathcal{P}_+$.} \qed

\vspace{1em}\noindent\textbf{Proof of Theorem \ref{thm:bp_equiv}.}

Recall the basis CDF over the ensemble of path-indexed bases:

\begin{equation*}
    \text{B-CDF}_D(x) = \frac{ \sum_{p \in \mathcal{P}} \sum_{C \in V(\mathcal{B}_p)} w(C) \mathbb{I}(\text{MPI}_D(C) \le x) }{ \sum_{p \in \mathcal{P}} \sum_{C \in V(\mathcal{B}_p)} w(C) }.
\end{equation*}

Following the logic of Theorem \ref{thm:ba_equiv}, substituting $w(C) = 1/f(n, |C|)$ into the numerator yields:

\begin{align*}
    \sum_{p \in \mathcal{P}} \sum_{C \in V(\mathcal{B}_p)} w(C) \mathbb{I}(\text{MPI}_D(C) \le x) 
    &= \sum_{C \in V} \left( \sum_{p \in \mathcal{P}} \mathbb{I}(C \in \mathcal{B}_p) \right) \frac{\mathbb{I}(\text{MPI}_D(C) \le x)}{f(n, |C|)} \\[0.85em]
    &= \sum_{C \in V} \mathbb{I}(\text{MPI}_D(C) \le x).
\end{align*}

As in the proof of Theorem \ref{thm:ba_equiv}, the denominator reduces identically to $\sum_{C \in V} 1 = |V|$; since $D$ violates GARP ($|V| \ge 1$), the ratio is strictly well-defined:
\begin{equation*}
    \text{B-CDF}_D(x) = \frac{1}{|V|} \sum_{C \in V} \mathbb{I}(\text{MPI}_D(C) \le x) = \text{O-CDF}_D(x) \quad \text{for all } x \in \mathbb{R}.
\end{equation*}
As the two CDFs are identically equal point-by-point for all $x \in \mathbb{R}$, their generalized inverses must coincide for any $\alpha \in (0, 1)$:
\begin{equation*}
    \text{BP}_\alpha\text{-MPI}_D = \min \{x : \text{B-CDF}_D(x) \ge \alpha\} = \min \{x : \text{O-CDF}_D(x) \ge \alpha\} = \text{OP}_\alpha\text{-MPI}_D,
\end{equation*}
which completes the proof. \qed

\vspace{1em}\noindent\textbf{Proof of Proposition \ref{prop:percentile_MPI}.}

Let $\mathcal{X}$ denote the finite support of cycle MPIs. For each $x \in \mathcal{X}$, the empirical basis CDF is a ratio of sample means defined as:
\begin{equation*}
    \text{B-CDF}_{D,K}(x) = \frac{\frac{1}{K}\sum_{i=1}^K \sum_{C \in V(\mathcal{B}_{p_i})} w(C) \mathbb{I}(\text{MPI}_D(C) \le x)}{\frac{1}{K}\sum_{i=1}^K \sum_{C \in V(\mathcal{B}_{p_i})} w(C)} = \frac{\bar{Y}_K(x)}{\bar{W}_K}.
\end{equation*}
By the SLLN and the CMT, $\text{B-CDF}_{D,K}(x) \xrightarrow{\text{a.s.}} \text{B-CDF}_D(x)$ for each $x \in \mathcal{X}$. As the support $\mathcal{X}$ is finite, pointwise almost sure convergence implies uniform almost sure convergence:
\begin{equation*}
    \sup_{x \in \mathcal{X}} |\text{B-CDF}_{D,K}(x) - \text{B-CDF}_D(x)| \xrightarrow{\text{a.s.}} 0.
\end{equation*}
Having established uniform convergence of the empirical CDF, we now establish the asymptotic behavior of the sample percentile $\text{BP}_\alpha\text{-MPI}_{D,K} = \min \{x \in \mathcal{X} : \text{B-CDF}_{D,K}(x) \ge \alpha\}$ by looking at the following two cases, where either $\alpha$ intersects $\text{B-CDF}_D(x)$ at a vertical jump or along a flat horizontal step.

\noindent\textit{Case (i):} ($\text{B-CDF}_D(M_\alpha^-) > \alpha$).
Here $M_\alpha^- = M_\alpha^+ = M_\alpha$. With the convention $\max \emptyset = 0$, define:
\begin{equation*}
    \delta = \min \left( \alpha - \max_{x < M_\alpha} \text{B-CDF}_D(x), \; \text{B-CDF}_D(M_\alpha) - \alpha \right) > 0.
\end{equation*}
By uniform convergence, there exists almost surely a finite $K_0$ such that for all $K \ge K_0$:
\begin{equation*}
    \sup_{x \in \mathcal{X}} |\text{B-CDF}_{D,K}(x) - \text{B-CDF}_D(x)| < \delta.
\end{equation*}
This implies that for all $K \ge K_0$ and every $x \in \mathcal{X}$,
\begin{equation*}
    \text{B-CDF}_D(x) - \delta < \text{B-CDF}_{D,K}(x) < \text{B-CDF}_D(x) + \delta.
\end{equation*}
Evaluating the lower bound at $x = M_\alpha$ and taking the maximum of the upper bound over $x < M_\alpha$, it follows directly from the definition of $\delta$ that for all $K \ge K_0$:
\begin{align*}
    \text{B-CDF}_{D,K}(M_\alpha) &> \text{B-CDF}_D(M_\alpha) - \delta \ge \alpha, \\[0.85em]
    \max_{x < M_\alpha} \text{B-CDF}_{D,K}(x) &< \max_{x < M_\alpha} \text{B-CDF}_D(x) + \delta \le \alpha.
\end{align*}
Thus, $\text{B-CDF}_{D,K}(x) < \alpha$ for all $x < M_\alpha$ and $\text{B-CDF}_{D,K}(M_\alpha) \ge \alpha$. By definition of the sample percentile, $\text{BP}_\alpha\text{-MPI}_{D,K} = \min \{x \in \mathcal{X} : \text{B-CDF}_{D,K}(x) \ge \alpha\} = M_\alpha$ for all $K \ge K_0$, establishing $\text{BP}_\alpha\text{-MPI}_{D,K} \xrightarrow{\text{a.s.}} M_\alpha$.

\vspace{0.5em}
\noindent\textit{Case (ii):} ($\text{B-CDF}_D(M_\alpha^-) = \alpha$).
When $\alpha$ coincides with a step height, $M_\alpha^- < M_\alpha^+$. Let:
\begin{equation*}
    \delta_0 = \min \left( \alpha - \max_{x < M_\alpha^-} \text{B-CDF}_D(x), \; \min_{x \ge M_\alpha^+} \text{B-CDF}_D(x) - \alpha \right) > 0.
\end{equation*}
By uniform convergence, there exists almost surely a finite $K_1$ such that for all $K \ge K_1$:
\begin{align*}
    \max_{x < M_\alpha^-} \text{B-CDF}_{D,K}(x) &< \alpha - \delta_0 + \delta_0 = \alpha, \\[0.85em]
    \min_{x \ge M_\alpha^+} \text{B-CDF}_{D,K}(x) &> \alpha + \delta_0 - \delta_0 = \alpha.
\end{align*}
As $\text{B-CDF}_{D,K}(x) < \alpha$ for all $x < M_\alpha^-$, the sample percentile cannot be strictly less than $M_\alpha^-$. At $M_\alpha^-$, if $\text{B-CDF}_{D,K}(M_\alpha^-) \ge \alpha$, the sample percentile is $M_\alpha^-$; otherwise, if $\text{B-CDF}_{D,K}(M_\alpha^-) < \alpha$, the next support point in $\mathcal{X}$ is $M_\alpha^+$, where $\text{B-CDF}_{D,K}(M_\alpha^+) > \alpha$ guarantees that the sample percentile is $M_\alpha^+$. Thus:
\begin{equation*}
    \text{BP}_\alpha\text{-MPI}_{D,K} \in \{M_\alpha^-, M_\alpha^+\} \quad \text{for all } K \ge K_1,
\end{equation*}
which holds almost surely and completes the proof. \qed

\vspace{1em}\noindent\textbf{Proof of Proposition \ref{prop:percentile_bias}.}

Let $\text{B-CDF}_D(x) = \frac{\mu_Y(x)}{\mu_W}$, where $\mu_Y(x) = \mathbb{E}[Y_i(x)]$ and $\mu_W = \mathbb{E}[W_i]$. Over $K$ i.i.d. sampled paths $(p_1, \dots, p_K)$, let $Y_i(x) = \sum_{C \in V(\mathcal{B}_{p_i})} w(C) \mathbb{I}(\text{MPI}_D(C) \le x)$ and $W_i = \sum_{C \in V(\mathcal{B}_{p_i})} w(C)$. The empirical basis CDF is $\text{B-CDF}_{D,K}(x) = \frac{\bar{Y}(x)}{\bar{W}}$, where $\bar{Y}(x) = \frac{1}{K}\sum_{i=1}^K Y_i(x)$ and $\bar{W} = \frac{1}{K}\sum_{i=1}^K W_i$.

For each path $p_i$, define the random variable:
\begin{equation*}
    Z_i(x) = Y_i(x) - \text{B-CDF}_D(x) W_i,
\end{equation*}
which is zero in expectation. As $0 \le Y_i(x) \le W_i \le W_{\max}$ and $\text{B-CDF}_D(x) \in [0, 1]$, the range of $Z_i(x)$ is $[- \text{B-CDF}_D(x) W_{\max}, \; (1 - \text{B-CDF}_D(x)) W_{\max}]$, which has length at most $W_{\max}$. Define the normalized variable:
\begin{equation*}
    U_i(x) = \frac{Z_i(x)}{W_{\max}} + \text{B-CDF}_D(x).
\end{equation*}
Then $U_i(x) \in [0, 1]$ and $\mathbb{E}[U_i(x)] = \text{B-CDF}_D(x)$.

As $W_{\min} = \min_{p \in \mathcal{P}} W(p) > 0$ and for every sampled path $W_i \ge W_{\min}$, then $\bar{W} \ge W_{\min} > 0$. Consider a deviation of size $\delta > 0$:
\begin{align*}
    |\text{B-CDF}_{D,K}(x) - \text{B-CDF}_D(x)| \ge \delta 
    &\iff \left| \frac{\bar{Y}(x)}{\bar{W}} - \text{B-CDF}_D(x) \right| \ge \delta \\[0.85em]
    &\iff |\bar{Z}(x)| \ge \delta \bar{W} \\[0.85em]
    &\implies |\bar{Z}(x)| \ge \delta W_{\min} \\[0.85em]
    &\iff \left| \frac{\bar{Z}(x)}{W_{\max}} \right| \ge \delta \frac{W_{\min}}{W_{\max}} \\[0.85em]
    &\iff |\bar{U}(x) - \text{B-CDF}_D(x)| \ge \delta \frac{W_{\min}}{W_{\max}}.
\end{align*}

As $U_i(x)$ for all sample paths are i.i.d. random variables in $[0, 1]$, Hoeffding's inequality\footnote{\citet[Theorem 1]{hoeffding1963probability}.} applied to $S_K = \sum_{i=1}^K U_i(x) = K \bar{U}(x)$ with deviation $\epsilon = K \delta \frac{W_{\min}}{W_{\max}}$ yields:
\begin{equation}
    \mathbb{P}\left(|\text{B-CDF}_{D,K}(x) - \text{B-CDF}_D(x)| \ge \delta\right) \le 2 \exp\left( - \frac{2 K \delta^2 W_{\min}^2}{W_{\max}^2} \right). \label{eq:hoeff_prob}
\end{equation}

From the proof of Proposition \ref{prop:percentile_MPI}, we know $\text{BP}_\alpha\text{-MPI}_{D,K} = M_\alpha$ if and only if $\text{B-CDF}_{D,K}(M_\alpha) \ge \alpha$ and $\text{B-CDF}_{D,K}(x) < \alpha$ for all $x < M_\alpha$. 

If $\text{B-CDF}_{D,K}(M_\alpha) < \alpha$, then since $\text{B-CDF}_D(M_\alpha) \ge \alpha + \delta$ (by definition of $\delta$), this requires $|\text{B-CDF}_{D,K}(M_\alpha) - \text{B-CDF}_D(M_\alpha)| \ge \delta$. This occurs with probability bounded by Equation \eqref{eq:hoeff_prob} with $x = M_\alpha$.

Consider the case where $\text{B-CDF}_{D,K}(x) \ge \alpha$ for some $x < M_\alpha$. If $M_\alpha = \min \mathcal{X}$, then $\text{B-CDF}_{D,K}(x) = 0 < \alpha$ for all $x < M_\alpha$, so this case is not possible. If $M_\alpha > \min \mathcal{X}$, let $x^* = \max \{x \in \mathcal{X} \mid x < M_\alpha\}$. As $\text{B-CDF}_{D,K}(x)$ is non-decreasing, $\text{B-CDF}_{D,K}(x) \ge \alpha$ for some $x < M_\alpha$ if and only if $\text{B-CDF}_{D,K}(x^*) \ge \alpha$. Since $\text{B-CDF}_D(x^*) \le \alpha - \delta$ (by definition of $\delta$), this requires $|\text{B-CDF}_{D,K}(x^*) - \text{B-CDF}_D(x^*)| \ge \delta$. This occurs with probability bounded by Equation \eqref{eq:hoeff_prob} with $x = x^*$.

Thus, the (union) probability that $\text{BP}_\alpha\text{-MPI}_{D,K} \neq M_\alpha$ is bounded from above by:

\begin{equation*}
    \mathbb{P}\left(\text{BP}_\alpha\text{-MPI}_{D,K} \neq M_\alpha\right) \le 4 \exp\left( - \frac{2 K \delta^2 W_{\min}^2}{W_{\max}^2} \right).
\end{equation*}

Finally, since both $\text{BP}_\alpha\text{-MPI}_{D,K} \in [0, 1]$ and $M_\alpha \in [0, 1]$, the absolute error cannot exceed $1$ when $\text{BP}_\alpha\text{-MPI}_{D,K} \neq M_\alpha$ and equals $0$ otherwise. Thus:

\begin{equation*}
    \left| \mathbb{E}[\text{BP}_\alpha\text{-MPI}_{D,K}] - M_\alpha \right| \le \min\left( 1, \; 4 \exp\left( - \frac{2 K \delta^2 W_{\min}^2}{W_{\max}^2} \right) \right),
\end{equation*}
which completes the proof. \qed
\end{appendices}

\section*{Declarations}
\begin{itemize}
\item \textbf{Funding:} No funds, grants, or other support was received.
\item \textbf{Competing Interests:} The author has no competing interests to declare that are relevant to the content of this article.
\item \textbf{Data Availability:} The Stanford Basket scanner dataset analyzed during the current study is publicly available from the replication repository of \citet{echenique2011money} and is also available on \href{https://github.com/GavinKader/Tractable-Estimation-of-the-Money-Pump-Index}{GitHub}.
\item \textbf{Code Availability:} All replication code, tables, and compiled figures are available on \href{https://github.com/GavinKader/Tractable-Estimation-of-the-Money-Pump-Index}{GitHub}.
\end{itemize}

\clearpage

\bibliography{references}

\end{document}